\documentclass[letterpaper]{article}
\usepackage[T1]{fontenc}
\usepackage[latin9]{inputenc}
\usepackage{geometry}
\usepackage{hyperref}

\usepackage{calc}
\usepackage{amsmath,amsthm,mathrsfs}
\usepackage{cleveref}
\usepackage{thmtools}

\usepackage{amssymb}
\usepackage{graphicx}
\usepackage{color, xcolor}
\usepackage{algorithm,algorithmic}
\usepackage{rotating}
\usepackage{xspace}
\usepackage[numbers]{natbib}

\makeatletter

\newcommand{\goodprivacy}{{induced subspace differentially private}\xspace}
\newcommand{\goodprivacynoun}{{induced subspace differential privacy}\xspace}

\newcommand{\fang}[1]{{\color{green!70!black}{ \bf \sf \scriptsize Fang:} \sf \scriptsize #1}}
\newcommand{\jie}[1]{{\color{blue!70!black}{ \bf \sf \scriptsize Jie:} \sf \scriptsize #1}}
\newcommand{\robin}[1]{{\color{red!70!black}{ \bf \sf \scriptsize Robin:} \sf \scriptsize #1}}
\renewcommand{\fang}[1]{}
\renewcommand{\jie}[1]{}
\renewcommand{\robin}[1]{}

\newtheorem{theorem}{Theorem}
\newtheorem{lemma}[theorem]{Lemma}
\newtheorem{proposition}[theorem]{Proposition}
\newtheorem{corollary}[theorem]{Corollary}

\theoremstyle{definition}
\newtheorem{definition}[theorem]{Definition}

\newtheorem{example}[theorem]{Example}

\newcommand{\A}{{\mathcal{A}}}
\newcommand{\x}{{\mathbf{x}}}
\DeclareMathOperator{\hist}{hist}
\DeclareMathOperator{\err}{err}
\DeclareMathOperator{\opt}{opt}
\DeclareMathOperator{\E}{{\mathbb{E}}}
\DeclareMathOperator{\Lap}{Lap}

\makeatother

\date{}

\begin{document}

\title{Subspace Differential Privacy\footnote{A short version of this paper is appeared in AAAI'22. The authors gratefully acknowledge research support from the National Science Foundation (OAC-1939459, CCF-2118953, CCF-1934924, DMS-1916002, and ISS-2007887).}}
\author{Jie Gao\thanks{Rutgers University, \texttt{jg1555@cs.rutgers.edu}}\and Ruobin Gong\thanks{ Rutgers University, \texttt{rg915@stat.rutgers.edu}}\and Fang-Yi Yu\thanks{Harvard University, \texttt{fangyiyu@seas.harvard.edu}}}
\maketitle
\begin{abstract}
    
Many data applications have certain invariant constraints due to practical needs. Data curators who employ differential privacy need to respect such constraints on the sanitized data product as a primary utility requirement. Invariants challenge the formulation, implementation and interpretation of privacy guarantees. We propose \emph{subspace differential privacy}, to honestly characterize the dependence of the sanitized output on confidential aspects of the data. We discuss two design frameworks that convert well-known differentially private mechanisms, such as the Gaussian and the Laplace mechanisms, to subspace differentially private ones that respect the invariants specified by the curator. For linear queries, we discuss the design of near optimal mechanisms that minimize the mean squared error. Subspace differentially private mechanisms rid the need for post-processing due to invariants, preserve transparency and statistical intelligibility of the output, and can be suitable for distributed implementation. We showcase the proposed mechanisms on the 2020 Census Disclosure Avoidance demonstration data, and a spatio-temporal dataset of mobile access point connections on a large university campus.

\end{abstract}

\section{Introduction}\label{sec:intro}




\noindent\textbf{Invariants: a  challenge for data privacy}
Data publication that satisfies differential privacy carries the formal mathematical guarantee that an adversary cannot effectively tell the difference, in the probabilistic sense, when two databases differ in only one entry. 
The extent of privacy protection under differential privacy is quantified by the privacy loss budget parameters, such as in $\epsilon$-differential privacy and $(\epsilon, \delta)$-differential privacy \citep{dwork2006calibrating}.
While the differential privacy guarantee is rigorously formulated with the probability language, its construction does not naturally mingle with hard and truthful constraints,  called \emph{invariants} \cite{ashmead2019effective},  that need to be imposed onto the sanitized data product, often as a primary utility requirement.

An important use case in which the challenge for privacy arises from invariants is the new Disclosure Avoidance System (DAS) of the 2020 Decennial Census \cite{abowd2018us}. The new DAS tabulates noise-infused, differentially private counts into multi-way contingency tables at various geographic resolutions, from the aggregated state and county levels to specific Census blocks. Due to the Census Bureau's constitutional mandate and its responsibilities as the official statistics agency of the United States, all data products must be preserved in such a way that certain aspects of their values are \emph{exactly} as enumerated. These invariants include (and are not limited to) population totals at the state level, counts of total housing units, as well as other group quarter facilities at the block level. Straightforward tabulations of the noisy measurements are most likely inconsistent with the mandated invariants. 

A common method to impose invariants on a differentially private noisy query is via post-processing using distance minimization \cite{abowd2019census,ashmead2019effective}.
 The resulting query is the solution to an optimization task, one that minimizes a pre-specified distance between the unconstrained query and the invariant-compliant space.  There are two major drawbacks to this approach. First, post-processing may introduce systematic bias into the query output. Particularly troubling is that the source of such bias is poorly understood \cite{DBLP:journals/corr/abs-2010-04327}, in part due to the highly data-dependent nature of the post-processing procedure, and the lack of a transparent probabilistic description. The TopDown algorithm \cite{abowd2019census}, employed by the Census DAS to impose invariants on the noisy measurements, exhibits a notable bias that it tends to associate larger counts with positive errors, whereas smaller counts with negative errors, when the total count is held as invariant. Figure~\ref{fig:census_trend} illustrates this phenomenon using the November 2020 vintage Census demonstration data \cite{ipums2020das}. For all states and state-level territories with more than five counties   (i.e. excluding D.C., Delaware, Hawaii and Rhode Island), a simple regression is performed between the county-level DAS errors and the log true county population sizes. Of the 48 regressions, 37 result in a negative slope estimate, out of which 11 are statistically significant at $\alpha = 0.01$ level (red circles in the left panel), indicating a systematic negative association between the DAS errors and the true counts. The bias trend is clearly visible in the right panel among the DAS errors (red squares) associated with the counties of Illinois, ordered by increasing population size. As \citet{DBLP:journals/corr/abs-2010-04327} discussed, the bias exhibited in the demonstration data is attributed to the non-negativity constraints imposed on the privatized counts.
 

The second, and more fundamental, drawback is that imposing invariants that are true aspects of the confidential data may incur additional privacy leakage, even if the invariants are also publicly released \cite{gong2020congenial}. When conveyed to data users and contributors, the narrative that the invariants are imposed by ``post-processing'' may lend to the erroneous interpretation that no additional leakage would occur. Care needs to be taken to explain that whenever nontrivial invariants are observed, the usual  $(\epsilon,\delta)$-differential privacy guarantee cannot be understood in its fullest sense.


The need to impose invariants arises in application areas involving the monitoring of stochastic events in structured spatio-temporal databases.
In some cases, there are requirements to report accurate counts (service requests \cite{nyc2021} or traffic volumes \cite{sui2016characterizing,yang2019revisiting,yang2020lbsn2vec,wang2020distributed}). In other cases, there are invariants that can be derived from external common sense knowledge --e.g., the number of vehicles entering and leaving a tunnel should be the same (when there are no other exits or parking spaces). An adversary may potentially use such information to reverse engineer the privacy protection perturbation mechanisms~\cite{rezaei19privacy,rezaei21influencers}. 
Invariants pose a new challenge to both the curators and the users of private data products, prompting its recognition as a new source of compromise to privacy that stems from external mandates.




\begin{figure*}[t]
\centering
\includegraphics[width = .32\textwidth]{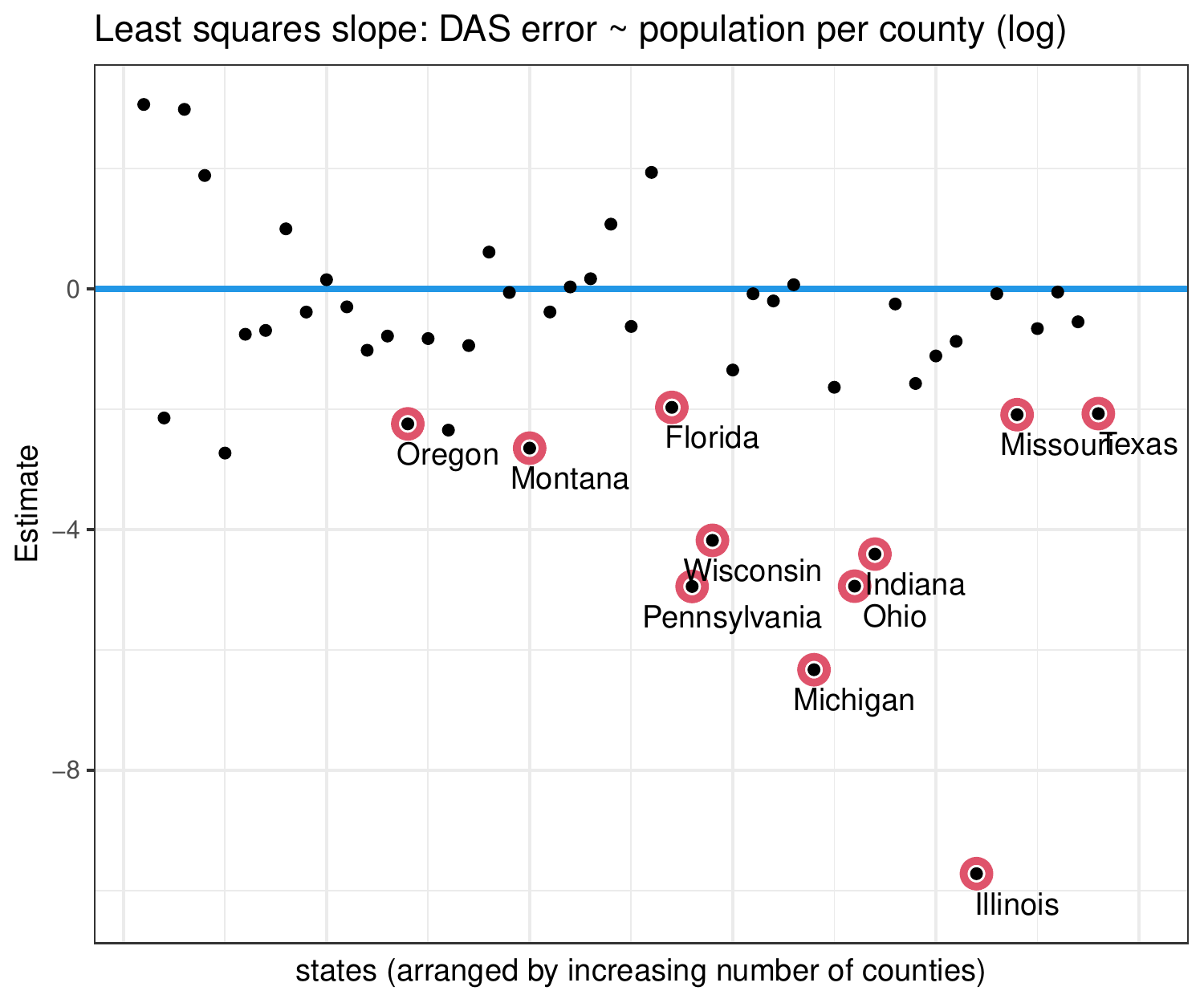}
\includegraphics[width = .55\textwidth]{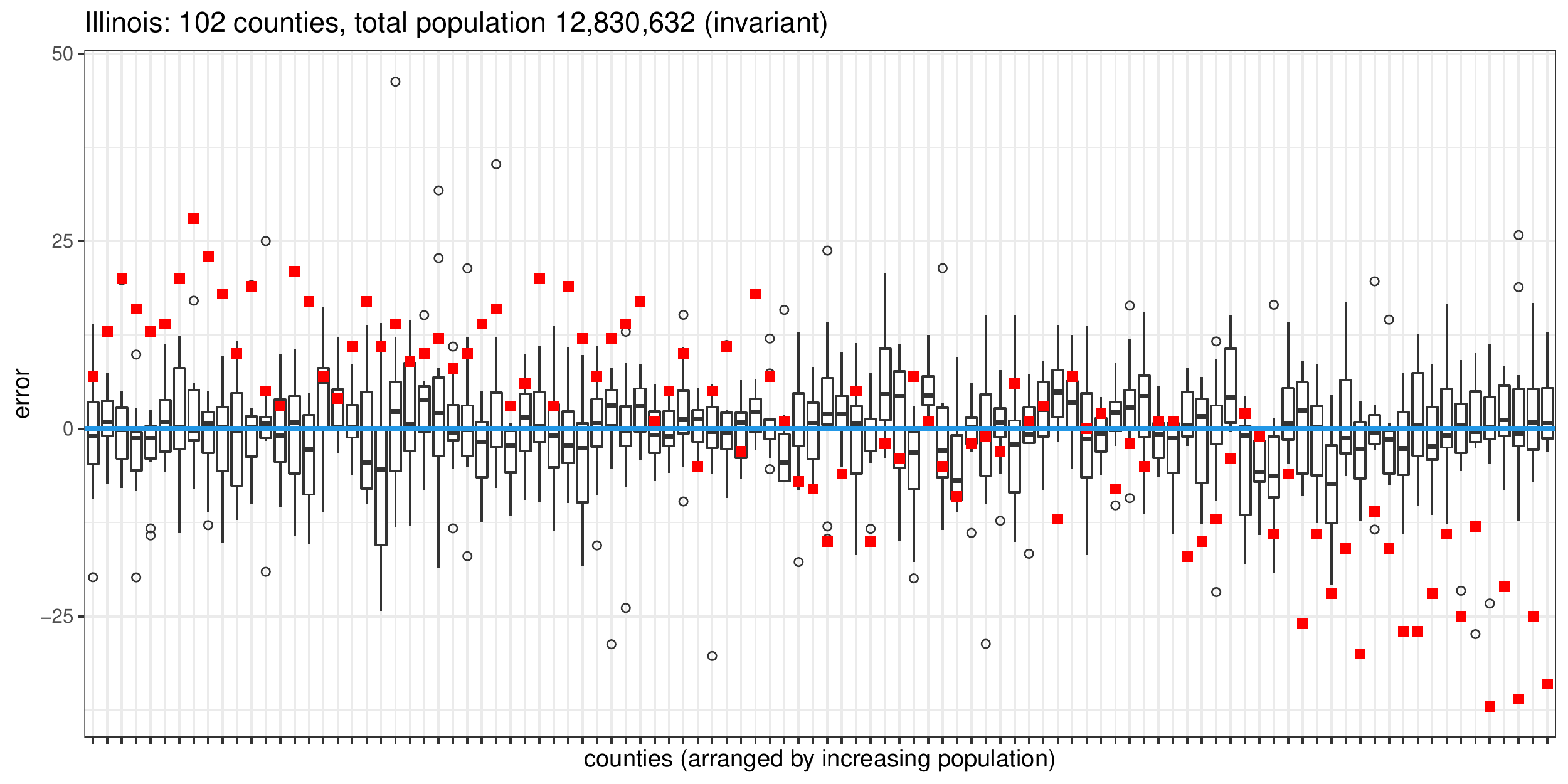}
\caption{\label{fig:census_trend} Left: the Census DAS associates positive errors with larger counties and negative errors with smaller counties, when the state population is held as invariant \cite[Nov 2020 vintage demonstration files;][]{ipums2020das}. 
Eleven out of 48 simple regressions of the county-level DAS errors against log true county populations have statistically significant negative slopes ($\alpha = 0.01$), circled in red. Right: for the counties of Illinois in increasing true population sizes, DAS errors (red squares) show a clear negative trend bias. The boxplots show errors from ten runs of our proposed method (the $(\epsilon,0)$-induced subspace differentially private projected Laplace mechanism; \Cref{cor:projected_laplace}). As  Corollary~\ref{cor:unbiasedness} shows, these errors are unbiased.}

\end{figure*}

\noindent\textbf{Our contribution}
To meet the challenge posed by invariants, we argue that the definition of differential privacy must be recapitulated to respect the given constraints. To this end, we propose the definition of \emph{subspace differential privacy}, which makes explicit the invariants imposed on the data product. Subspace differential privacy intends to honestly characterize the dependence of the privatized data product on truthful aspects of the confidential data, to ensure that any privacy guarantee attached to the data product is both mathematically rigorous and intuitively sensible. It enables the assessment of what kind of queries do, and do not, enjoy the protection of differential privacy. 

The literature has seen attempts to generalize beyond the classic notion of differential privacy. The framework of {Pufferfish} privacy~\cite{kifer2012rigorous,kifer2014pufferfish,song2017pufferfish} specifies the potential secrets, discriminative pairs, as well as the data generation model and knowledge possessed by the potential attacker. Special cases of the Pufferfish framework include {Blowfish} privacy~\cite{He2014-id} and Bayesian differential privacy~\cite{Yang2015-bc}. Related notions of correlated differential privacy \citep{zhu2014correlated} and dependent differential privacy \citep{liu2016dependence} specifically address secrets in the form of query correlations or structural dependence of the database. 

The current work differentiates itself from the existing literature in two senses. First, the theoretical focus is to provide a principled reconciliation between the hard truth constraints which the data curator must impose on the sanitized data product, and the privacy guarantee the product can enjoy. In particular, just like the classic notion of differential privacy, \emph{subspace differential privacy} does not require the specification of a data generation model nor any knowledge that the attacker might possess. Second, the practical emphasis is on the design of probabilistic mechanisms that impose deterministic truth constraints as they instill subspace differential privacy in the data product. This forgoes the need for additional post-processing, and preserves good statistical qualities of the output.

A related, but different, line of work in the literature concerns the \emph{internal} consistency of the privacy mechanism output~\cite{barak2007privacy,hay2009boosting}. For example, when we query the number of students in a school and the numbers of students in each class of the school, we may expect the outputs to be non-negative, and the sum of the (privatized) numbers of students in all classes to be equal to the (privatized) number of students in the school. These internal consistency requirements, such as non-negativity and relational constraints, are \emph{independent} of the private dataset. Therefore, they may be compatible with the classic notion of differential privacy, in which case they may be instantiated with differentially private mechanisms. However, for invariants that are nontrivial functions of the confidential data, we show in Section~\ref{sect:pre_priv} that it is impossible to have differentially private mechanisms that satisfy them. It is this kind of invariants that motivate our work in this paper.

The remainder of this paper is organized as follows. Section~\ref{sec:sdp} defines \emph{subspace differential privacy} and \emph{induced subspace differential privacy}, motivated by the pair of necessary criteria that the mechanism be simultaneously provably private and invariant-respecting. Section~\ref{sec:design} outlines two general approaches, \emph{projection} and \emph{extension}, to design induced subspace differentially private mechanisms.  We apply both frameworks to produce Gaussian and Laplace mechanisms for general queries, present a correlated Gaussian mechanism that is  \emph{near-optimal} (i.e. in terms of mean squared error, with a small multiplicative factor) for linear queries, and sketch the design for a $k$-norm mechanism that would enjoy near optimality. Section~\ref{sec:discussion} discusses the statistical and implementation considerations behind the proposed mechanisms, as they enjoy transparency and statistical intelligibility that facilitate principled downstream statistical analysis. In the special case of additive spherical Gaussian mechanism, a distributional equivalence is established between the projection framework and statistical conditioning. All mechanisms can also be adapted for distributed privatization.  Section~\ref{sec:example} provides two demonstrations of the proposed induced subspace differentially private mechanisms, on the 2020 Census DAS demonstration data and spatio-temporal mobility dataset on a university campus subject to various marginal total invariants.  Section~\ref{sec:conclusion} concludes. 




\section{Recapitulating Privacy under Invariants}\label{sec:sdp}


In this work, we model private data as a database $\x=\left(x_{1},\dots,x_{N}\right)^{\top} \in  \mathcal{X}^N$ of $N$ rows, where each row $x_i\in \mathcal{X}$ contains data about an individual $i$, and $\mathcal{X}$ is finite with size $d$, and we set the space of all possible non-empty databases as $\mathcal{X}^*:=\cup_{N\ge 1}\mathcal{X}^N$.  
A trusted curator holds the database $\x\in \mathcal{X}^*$, and provides an interface to the database through a randomized mechanism $M:\mathcal{X}^*\to \mathcal{Y}$ where  $\mathcal{Y}\subseteq \mathbb{R}^n$ is the \emph{output space} of $M$.  We want to design good mechanisms to answer a query $A:\mathcal{X}^*\to \mathcal{Y}$ that satisfies not only certain privacy notions, but also invariant constraints as motivated in Section~\ref{sec:intro}.

We begin the discussion about mechanisms for a general query, defined by a function $A:\mathcal{X}^*\to \mathcal{Y}$, throughout the end of Section~\ref{sec:general}. In Section~\ref{sec:linear}, we consider optimal mechanisms for a \emph{linear} query $A$, with  $a:\mathcal{X}\to \mathbb{R}^n$ so that $A(\x) := \sum_{i} a(x_i)$.  Indeed, a linear query $A$ can be represented as a linear function of the \emph{histogram} of database $\x$, $\hist(\x):\mathcal{X}^*\to \mathbb{N}^{d}$ where $\hist(\x)_z := \sum_{i}\mathbf{1}[x_i = z]$ is the number of rows equal to $z\in \mathcal{X}$ in $\x\in \mathcal{X}^*$.  With this notation, given a linear query $A$, we denote $\A$ as a matrix where the $k,z$ entry is $A(z)_k$ for $k\in [n]$ and $z\in \mathcal{X}$, and the linear query on a database $\x$ can be written as matrix multiplication, 
$A(\x) = \A \cdot \hist(\x)$.

\subsection{Privacy Guarantees and Invariants}\label{sect:pre_priv}

The notion of differential privacy ensures that no individual's data has much effect on the output of the mechanism $M$.  That is, if we consider any two neighboring databases $\x$ and $\x'$ of size $N$ that differ on one row (there exists $i$ such that $x_i \neq x_i'$ and  $x_j = x_j'$ for all $j\neq i$.) the output distribution of mechanism $M$ on $\x$ should be similar to that of $M$ on $\x'$.  Formally:
\begin{definition}[Bounded differential privacy~\cite{dwork2006calibrating}]\label{def:dp}
Let $(\mathcal{Y}, \mathcal{F})$ be a measurable space and $\epsilon, \delta\ge0$. We say that a random mechanism $M: \mathcal{X}^*\to\mathcal{Y}$ is \emph{$\left(\epsilon,\delta\right)$-differentially private} if for all neighboring databases $\x\sim \x'$ and all measurable set $S\in\mathcal{F}$, 
$\Pr\left[ M\left(\x\right)\in S\right]\le e^\epsilon\cdot \Pr\left[M\left(\x'\right)\in S\right]+\delta$.
\end{definition}

From the data curator's perpective, in addition to privacy concerns, there often exists external constraints that the privatized output $M$ must meet. 
These constraints can often be represented as a function of $M(\x)$ that agrees with what's calculated based on the confidential $A(\x)$. In this work, we focus on the class of invariants that can be written in the form of a system of linear equations. 
\begin{definition}[Invariants - linear equality]\label{def:invariant}
Given a query $A:\mathcal{X}^*\to \mathbb{R}^n$, and $C\in \mathbb{R}^{n_c\times n}$ be a $n_c\times n$ matrix with rank $n_c<n$.\footnote{We can always find a $C'$ consists of a subset of independent rows of $C$ which has same row rank as $C$'s,  and we can translate between these two through a linear transformation.} A (random) mechanism $M:\mathcal{X}^*\to \mathcal{Y}\subseteq \mathbb{R}^n$ satisfies the \emph{linear equality invariant $C$ with query $A$}, if for all $\x$,
$CM\left({\bf x}\right) =CA(\x)$ with probability one over the randomness of $M$. 
\end{definition}
Given a linear equality invariant $C$, let $\mathcal{N} :=\{v\in \mathbb{R}^n:Cv = 0\}$ be the \emph{null space} of $C$, and $\mathcal{R} := \mathcal{N}^\perp\subseteq\mathbb{R}^n$ be the \emph{row space} of $C$.  Additionally, we set $\Pi_\mathcal{N}$ be the orthogonal projection matrix for null space $\mathcal{N}$, $Q_\mathcal{N}$ be a collection of orthonormal basis of $\mathcal{N}$, and $A_\mathcal{N}:= Q_\mathcal{N}^\top A$ be the query function $A$ projected into $\mathcal{N}$ with basis $Q_\mathcal{N}$.  We use subscript $\mathcal{R}$ in the similar manner.  

Linear equality invariants are a natural family of invariants. Here are two examples. 
\begin{example}\label{ex:hist}
Let $\x\in \mathcal{X}^N$ be a database with $|\mathcal{X}| = d$, $A = \hist$ be the histogram query, and $C = \mathbf{1}^\top = (1,\dots,1)\in \mathbb{R}^{1\times d}$ all one vector with length $d$. 
This linear equality invariant requires the curator to accurately report the total number of individuals without error, because $CM(\x) = \mathbf{1}^\top \hist(\x) = N$.
\end{example}

\begin{example}\label{ex:even}
Consider $\mathcal{X} = \{1,2,3,4\}$, $A = \hist: \mathcal{X}^*\to \mathbb{R}^4$ and $C = (1, 0, 1, 0)$.  The linear equality invariant ensures the number of individual with odd feature is exact.
\end{example}

The invariants discussed in the Census DAS application, such as state-level population and block-level housing units and group quarter facilities, can be formulated as linear equality invariants.


\subsection{Subspace and Induced Subspace Differential Privacy}


The mechanism we seek must meet two necessary criteria: \emph{provably private} and \emph{invariant-respecting}. That is, the mechanism should privatize the confidential query with mathematically provable guarantees, while at the same time the query output should conform to the invariants that a data curator chooses to impose. The two criteria culminate in the \emph{\goodprivacynoun} (\Cref{def:isubdp}) which enjoys additional desirable properties, such as the practicalities of design and statistical intelligibility, which will be discussed from \cref{sec:design} and on.

To motivate the construction,  note that the classic definition of differential privacy and invariant constraints are not compatible by design.  For instance, in Example~\ref{ex:even} if a mechanism $M$ respects the invariant constraint~\Cref{def:invariant}, the probability ratio of event $S = \{(y_1, y_2, y_3, y_4): y_1+y_3 = 1\}\subset \mathbb{R}^4$ on  neighboring databases $\x = (1,2,4)$ and  $\x' = (4,2,4)$ is unbounded, $\frac{\Pr[M(\x)\in S]}{\Pr[M(\x')\in S]}=\infty$,
because $\Pr[M(\x)\in S] = \Pr[CM(\x) = C\hist(\x)] = 1$ but $\Pr[M(\x')\in S] = \Pr[CM(\x') \neq C\hist(\x')] = 0$.  Thus $M$ violates $(\epsilon, \delta)$-differential privacy for any $\epsilon>0$ and $\delta<1$.  

Therefore, we need a new notion of differential privacy to discuss the privacy protection in the presence of mandated invariants. Below we attempt to do so by recapitulating the definition of  $\left(\epsilon,\delta\right)$-differential privacy, to acknowledge the fact that if a hard linear constraint is imposed on the privacy mechanism, we can no longer offer differential privacy guarantee in the full $n$-dimensional space that is the image of $A$, but rather only within certain linear subspaces.

\begin{definition}[Subspace differential privacy]\label{def:subdp}
Let $\mathcal{V}$ be a linear subspace of $\mathbb{R}^{n}$, and $\Pi_{\mathcal{V}}$ the projection matrix onto $\mathcal{V}$. Given $\epsilon, \delta\ge 0$, a random mechanism $M:\mathcal{X}^*\to\mathbb{R}^{n}$ is \emph{$\mathcal{V}$-subspace $\left(\epsilon,\delta\right)$-differentially private} if for all neighboring databases ${\bf x}\sim {\bf x}'$ and { every Borel subset $S\subseteq \mathcal{V}$}, 
\begin{equation}
\Pr\left[ \Pi_{\mathcal{V}} M\left({\bf x}\right)\in S\right]\le e^\epsilon\Pr\left[\Pi_{\mathcal{V}}M\left({\bf x}'\right)\in S\right]+\delta.
\end{equation}
\end{definition}



We are ready to formalize the notion of a provably private and invariant-respecting private mechanism, one that meets both the criteria laid out at the beginning of this subsection.  
\begin{definition}[Induced subspace differential privacy]\label{def:isubdp}
Given $\epsilon, \delta\ge 0$, a query $A:\mathcal{X}^*\to \mathbb{R}^n$ and a linear equality invariant $C:\mathbb{R}^n\to \mathbb{R}^{n_c}$ with null space $\mathcal{N}$, a mechanism $M:\mathcal{X}^*\to \mathbb{R}^n$ is \emph{$(\epsilon, \delta)$-\goodprivacy for query $A$ and an invariant $C$} if 
1) $M$  is $\mathcal{N}$-subspace $(\epsilon, \delta)$-differentially private (\Cref{def:subdp}), and 
2) $M$ satisfies the linear equality invariant $C$ (\Cref{def:invariant}).
$\mathcal{M}$ may be referred to simply as \emph{\goodprivacy}, whenever the context is clear about (or does not require the specification of) $\epsilon, \delta$ and $C$.
\end{definition}
An \goodprivacy mechanism delivers query outputs that meet the curator's invariant specification ($C$) with probability one. It is provably differentially private for all queries and their components that are orthogonal to the invariants, and is silent on privacy properties for those that are linearly dependent on the invariants. 

\subsection{Properties of Subspace Differential Privacy} 

Now we discuss several properties of subspace differential privacy. First, we show a ``nestedness'' property:  a $\mathcal{V}_1$-subspace differentially private mechanism is also $\mathcal{V}_2$-subspace differentially private for all $\mathcal{V}_1\supseteq \mathcal{V}_2$.

\begin{proposition}\label{prop:nested}
Let ${\mathcal{V}}_2 \subseteq {\mathcal{V}_1}$ be nested linear subspaces of respective dimensions $d_2 \le d_1$. If a mechanism $M$ is ${\mathcal{V}_1}$-subspace $\left(\epsilon,\delta\right)$-differentially private, it is also ${\mathcal{V}}_2$-subspace $\left(\epsilon,\delta\right)$-differentially private.
\end{proposition}

The main idea is that because $\mathcal{V}_2\subseteq \mathcal{V}_1$ are both linear spaces, for any measurable $S$, we can always find another measurable set $S'$ such that $\Pi_{\mathcal{V}_2}^{-1}(S) = \Pi_{\mathcal{V}_1}^{-1} (S')$.  Thus $\mathcal{V}_1$-subspace differential privacy implies $\mathcal{V}_2$-subspace differential privacy.  We include the proof to the appendix.

\Cref{prop:nested} implies that a differentially private mechanism is subspace differential private, as shown in \Cref{prop:unconstrained}.  Thus, we can call an  $\left(\epsilon,\delta\right)$-differentially private mechanism the $\mathbb{R}^n$-subspace $\left(\epsilon,\delta\right)$-differentially private mechanism.   
\begin{corollary}[label = prop:unconstrained]
If $M:\mathcal{X}^*\to \mathbb{R}^n$ is a $\left(\epsilon,\delta\right)$-differentially private mechanism, it is $\mathcal{V}$-subspace $\left(\epsilon,\delta\right)$-differentially private for any linear subspace $\mathcal{V} \subseteq \mathbb{R}^n$.
\end{corollary}

Induced subspace differential privacy inherits the composition property from differential privacy in the following sense (with proof in Appendix~\ref{app:subspace}).
\begin{proposition}[name=Composition, label = prop:composition, restate=Composition]
Given $\epsilon_1\ge 0$ and $\epsilon_2\ge 0$, 
if $M_1$ is induced subspace $(\epsilon_1,0)$-differentially private for query $A_1$ and linear invariant $C_1$ and $M_2$ is induced subspace $(\epsilon_2,0)$-differentially private for query $A_2$ and linear invariant $C_2$, the composed mechanism $(M_1, M_2)$ is induced subspace $(\epsilon_1+\epsilon_2,0)$-differentially private for query $A_{1,2} := (A_1, A_2)$ and linear invariant $C_{1,2} := (C_1, C_2)$. 
\end{proposition}

The above is a preliminary answer to the composition property of subspace differential privacy. 
However, when $C_1$ and $C_2$ are different, the composition of invariant constraints may reveal additional information about the underlying confidential data. In general, the composition of logically independent  and informative invariants is not unlike a database linkage attack. 
For instance, $A_1 = A_2$ is a two-way contingent table that reports the counts of individuals with ages and zip code, the invariant $C_1$ ensures the accurate marginal counts of individuals within each age bracket, and $C_2$ ensures the accurate count of individuals within each zip code.  The composition of these two invariant constraints may allow an adversary to infer each individual's information.

The subspace differential privacy is naturally immune to any post-processing mapping which only acts on the subspace $\mathcal{N}$. However, it is not readily clear how to define post-processing under mandated invariant constraints.  Specifically, if a mechanism satisfies a nontrivial invariant constraint $C$ with row space $\mathcal{R}$ and null space $\mathcal{N}$, we can apply a post-processing mapping that outputs the invariant component in $\mathcal{R}$ -- here the output is revealed precisely. Unfortunately, such will be true for any invariant-respecting privatized product, regardless of what notion of privacy is attached to it.


Finally, we note that if we consider a mixture of two $\mathcal{V}$-subspace differentially privacy mechanisms the resulting mechanism is also $\mathcal{V}$-subspace differentially private.  This property is known as the privacy axiom of choice~\cite{kifer2010towards}.

\section{Mechanism Design}\label{sec:design}

This section introduces two frameworks for designing \goodprivacy mechanisms with linear equality invariant $C$. 
The data curator would invoke the \emph{projection} framework if seeking to impose invariants onto an \emph{existing} differentially private mechanism, and the \emph{extension} framework if augmenting a smaller private query in a manner compatible with the invariants. Both frameworks are applied to revise existing differentially private mechanisms in \Cref{sec:general}, notably the Gaussian and the Laplace mechanisms, for general queries. For linear queries, \Cref{sec:linear} presents a near optimal correlated Gaussian mechanism, and sketches the design of a near optimal $k$-norm mechanism.

\subsection{The Projection Framework}\label{sec:projection}

Suppose the data curator already employs a differentially private mechanism $M$ to answer a general query $A$, and would like to impose a linear equality invariant $C$ on the query output. The projection framework, outlined in~\Cref{thm:projection}, can transform the existing mechanism $M$ to \goodprivacy for $A$ and $C$, with little overhead on the curator's part. 

\begin{theorem}[Projection framework]\label{thm:projection}
Given $\epsilon, \delta\ge 0$, a general query $A:\mathcal{X}^*\to \mathbb{R}^n$, and a linear equality invariant $C$ with null space $\mathcal{N}$, if $M:\mathcal{X}^*\to \mathbb{R}^n$ is $(\epsilon, \delta)$-differentially private, then $\mathcal{M}(\x):=A(\x)+\Pi_\mathcal{N}(M(\x)-A(\x))$, for all $\x\in \mathcal{X}^*$ is $(\epsilon, \delta)$-\goodprivacy for query $A$ and invariant $C$.
\end{theorem}
Informally, we conduct projection on a differentially private mechanism in order to 1) remove the noise in the row space of $C$ to respect the invariant constraint $C$; and 2) preserve noise in the null space $\mathcal{N}$ and satisfy $\mathcal{N}$-subspace differential privacy. 

\begin{proof}[Proof of Theorem~\ref{thm:projection}]
Because for all $\x\in \mathcal{X}^*$, $C\mathcal{M}(\x)=C A(\x)+C\Pi_{\mathcal{N}}(M(\x)-A(\x)) = CA(\x)$, 
$\mathcal{M}$ satisfies the invariant $C$.
Since $M$ is $\mathcal{N}$-subspace $(\epsilon, \delta)$-differentially private by \Cref{prop:unconstrained}, and  $\Pi_{\mathcal{N}}\mathcal{M}(\x)$ equals $\Pi_{\mathcal{N}}M\left(\x\right)$ in distribution, $\mathcal{M}$ is also ${\mathcal{N}}$-subspace differentially private.
\end{proof}


The projection framework in \Cref{thm:projection} is particularly useful for revising \emph{additive} mechanisms, having the form
\begin{equation}\label{eq:additive}
M(\x) = A(\x)+\mathbf{e},	
\end{equation}
where ${\bf e}$ is a noise component independent of $A(\x)$. Examples of additive mechanisms include the classic Laplace and Gaussian mechanisms~\citep{dwork2006calibrating}, $t$-\citep{nissim2007smooth}, double Geometric \citep{fioretto2019differential}, and $k$-norm mechanisms~\citep{hardt2010geometry,bhaskara2012unconditional}. In contrast, the Exponential mechanism~\citep{mcsherry2007mechanism} is in general not additive because the sampling process depends on the utility function non-additively.  


When the existing differentially private mechanism $M$ is additive, the projection construction of an \goodprivacy mechanism based on $M$ can be simplified, by first sampling the noise $\mathbf{e}$, and outputting the query value $A(\x)$ with the projected noise added to it.
\begin{corollary}\label{cor:projection}
If the mechanism $M$ in Theorem~\ref{thm:projection} is furthermore additive, i.e. of the form~\eqref{eq:additive}, the modified mechanism,
$\mathcal{M}(\x):=A(\x)+\Pi_{\mathcal{N}}\mathbf{e}$,
is $(\epsilon, \delta)$-\goodprivacy for query $A$ and invariant $C$. 
\end{corollary}


\Cref{cor:projection} will be useful for the derivation of the projection mechanisms, as well as the various statistical properties of additive mechanisms (including the crucial \emph{unbiasedness} property), to be discussed in the ensuing sections.  

\subsection{The Extension Framework}\label{sec:extension}

The projection framework in Section~\ref{sec:projection} transforms an existing differentially private mechanism to one that is subspace differentially private and respects the invariant $C$. The extension framework, on the other hand, enables the design of an \goodprivacy mechanism without a full mechanism in place yet. \Cref{thm:connect} discusses how to extend a differential private mechanism with image contained in $\mathcal{N}$ to an \goodprivacy mechanism.  Moreover, the converse also holds--- any \goodprivacy mechanism can be written as a differential private mechanism with a translation. Thus, the extension framework provides the optimal trade-off between privacy and accuracy, as \Cref{cor:err_constrain2dp} will show.  We defer the proof to supplementary material.

\begin{theorem}[name = Extension framework,label = thm:connect,restate = connect]
Given $\epsilon, \delta\ge 0$, a general query $A:\mathcal{X}^*\to \mathbb{R}^n$, and a  linear equality invariant $C$ with null space $\mathcal{N}$ and row space $\mathcal{R}$, $\mathcal{M}$ is $(\epsilon, \delta)$-\goodprivacy for query $A$ and invariant $C$, if and only if 
$\mathcal{M}(\mathbf{x}):=\hat{M}(\mathbf{x})+\Pi_{\mathcal{R}} A(\x)$
where $\hat{M}$ is $(\epsilon, \delta)$-differentially private and its image is contained in $\mathcal{N}\subseteq \mathbb{R}^n$.
\end{theorem}

\subsection{Induced Subspace Differentially Private Mechanisms for General Queries}\label{sec:general}

We now describe the use of the above two frameworks to construct \goodprivacy mechanisms.  We first introduce \goodprivacy mechanisms with Gaussian noises, then the pure (i.e. $\delta = 0$) versions with Laplace noises.
All mechanisms discussed in this section are additive mechanisms, having a general functional form as \eqref{eq:additive}.

Let the \emph{$\ell_{p}$ sensitivity} of the query function $A:\mathcal{X}^*\to \mathbb{R}^n$ be
$\Delta_{p}(A) = \sup_{{\bf x}\sim {\bf x}'}\left\Vert A({\bf x}) - A({\bf x}') \right\Vert _{p}$, which  measures how much a single individual's data can change the output of the query $A$. We measure the performance of a mechanism $M$ in terms of the expected squared error.
Given any query function $A: \mathcal{X}^* \to \mathbb{R}^n$, the worst-case expected squared error of a mechanism $M$ for $A$ is defined as $\err_M(A):=\sup_{\x\in \mathcal{X}^*}\E\left[\|M(\x)-A(\x)\|^2_2\right]$.

\noindent\textbf{Gaussian mechanisms}
Recall the standard Gaussian mechanism, which adds a spherical noise to the output that depends on the $\ell_2$ sensitivity of $A$.   By an abuse of notation, the superscript $n$ over a probability distribution denotes the $n$-dimensional product distribution with the said marginal. 

\begin{lemma}[name=Gaussian mechanism~\cite{dwork2006our}, label=lem:gaussian]
For all $\epsilon, \delta>0$, and general query $A:\mathcal{X}^*\to \mathbb{R}^n$, let $c_{\epsilon, \delta}:= \epsilon^{-1}(1+\sqrt{1+\ln(1/\delta)})$.  
Then an additive mechanism for $A$ with noise $\mathbf{e}_{G}(A, \epsilon, \delta) \overset{d}{=} N(0;  \Delta_2(A)c_{\epsilon, \delta})^n$ where $N(0; \sigma)$ is the unbiased Gaussian distribution with variance $\sigma^2$ $(\epsilon, \delta)$-differentially private.
\end{lemma}

Given $\epsilon, \delta>0$, a general query $A:\mathcal{X}^*\to \mathbb{R}^n$, and linear equality invariant $C\in \mathbb{R}^{n_c\times n}$, by projection (\Cref{thm:projection}) and extension (\Cref{thm:connect}), we can derive two \goodprivacy Gaussian mechanisms.  
The proofs of Corollaries~\ref{cor:projected_gaussian} and~\ref{cor:extended_gaussian} are both given in Appendix~\ref{app:general}.  


\begin{corollary}[name=Projected Gaussian mechanism, label=cor:projected_gaussian]
An additive mechanism $\mathcal{M}_{PG}$ for $A$ with noise $\Pi_\mathcal{N}\mathbf{e}_{G}(A, \epsilon, \delta)$ where $\mathbf{e}_{G}$ is defined in \Cref{lem:gaussian} is $(\epsilon, \delta)$-\goodprivacy for query $A$ and invariant $C$. 
Moreover, 
$\err_{\mathcal{M}_{PG}}(A) = (n-n_c)c_{\epsilon, \delta}^2\Delta_2(A)^2$.
\end{corollary}


To apply the extension framework of~\Cref{thm:connect},  one complication is how to design a differentially private mechanism with image in $\mathcal{N}$.  To handle this, we first project the query $A$ to the null space $\mathcal{N}$ and have a new query $A_\mathcal{N} = Q_\mathcal{N}^\top A: \mathcal{X}^*\to \mathbb{R}^{n-n_c}$. Then, compute the sensitivity of $A_\mathcal{N}$, $\Delta_2(A_\mathcal{N})$, and sample $\mathbf{e}_\mathcal{N}$ which consists of $n-n_c$ iid Gaussian noise with variance $(c_{\epsilon, \delta}\Delta_2(A_\mathcal{N}))^2$.  Finally, convert the noise to the original space $\mathbb{R}^n$ and add the true query $A(\x)$.  We define the mechanism formally below.
\begin{algorithm}[htb]
\caption{Gaussian \goodprivacy mechanism through extension}
\label{alg:gaussian_extended}
\textbf{Input}: a database $\x$, a query $A:\mathcal{X}^*\to \mathbb{R}^n$, linear equality invariant $C\in \mathbb{R}^{n_c\times n}$ with rank $n_c$, $\epsilon\in (0,1)$, and $\delta\in (0,1)$.
\begin{algorithmic}[1] 
\STATE Compute $Q_\mathcal{N}\in \mathbb{R}^{n\times n-n_c}$ an collection of an orthonormal basis of $\mathcal{N}$, and $A_\mathcal{N}:= Q^\top_\mathcal{N}A$.
\STATE Let $c_{\epsilon, \delta} = \epsilon^{-1}(1+\sqrt{1+\ln(1/\delta)})$, and sample $\mathbf{e}_\mathcal{N}\overset{d}{=} N\left(0;c_{\epsilon, \delta}\Delta_2(A_\mathcal{N})\right)^{n-n_c}.$
\STATE \textbf{return} $A(\x)+Q_\mathcal{N} \mathbf{e}_\mathcal{N}$.
\end{algorithmic}
\end{algorithm}

\begin{corollary}[name=Extended Gaussian mechanism, label=cor:extended_gaussian]
An additive mechanism $\mathcal{M}_{EG}$ for $A$ with noise $\mathbf{e}_{EG}(A, \epsilon, \delta) = Q_\mathcal{N}e_\mathcal{N}$ where $e_\mathcal{N}\overset{d}{=} N\left(0;c_{\epsilon, \delta}\Delta_2(Q_\mathcal{N}^\top A)\right)^{n-n_c}$  is $(\epsilon, \delta)$-\goodprivacy for query $A$ and invariant $C$. 
Moreover, 
$\err_{\mathcal{M}_{EG}}(A) = (n-n_c)c_{\epsilon, \delta}^2\Delta_2(Q_\mathcal{N}^\top A)^2$.
\end{corollary}

Since $Q_\mathcal{N}$ consists of orthonormal columns, the $\ell_2$ sensitivity of $Q_\mathcal{N}^\top A$ is less than or equal to the sensitivity of the original query $A$, and the error in \Cref{cor:extended_gaussian} is no more than the error in \Cref{cor:projected_gaussian}.

\noindent\textbf{Laplace mechanisms}
Similarly, the standard Laplace mechanism adds independent product Laplace noise to the output that depends on the $\ell_1$ sensitivity of $A$. In what follows, $\Lap(b)$ denotes the univariate Laplace distribution with scale $b>0$, with density function $\frac{1}{2b}\exp\left(-\frac{|x|}{b}\right)$.
\begin{lemma}[name=Laplace mechanism~\citep{dwork2006calibrating}, label=lem:laplace]
Given $\epsilon>0$, and a query $A:\mathcal{X}^*\to \mathbb{R}^n$, an additive mechanism for $A$ with noise $\mathbf{e}_L(A, \epsilon)\overset{d}{=} \Lap(  \Delta_1(A)/\epsilon)^n$ is $(\epsilon, 0)$-differentially private.
\end{lemma}
 

Given $\epsilon>0$, a query $A:\mathcal{X}^*\to \mathbb{R}^n$, and a linear equality invariant $C\in \mathbb{R}^{n_c\times n}$, we have the following two \goodprivacy Laplace mechanisms.
\begin{corollary}[name=Projected Laplace mechanism, label=cor:projected_laplace]
The additive mechanism for $A$ with noise $\Pi_\mathcal{N}\mathbf{e}_{L}$ where $\mathbf{e}_{L}$ is defined in \Cref{lem:laplace} is $(\epsilon, 0)$-\goodprivacy for query $A$ and invariant $C$. 
\end{corollary}

\begin{corollary}[name=Extended Laplace mechanism, label=cor:extended_laplace]
An additive mechanism $\mathcal{M}_{LE}$ for $A$ with noise $\mathbf{e}_{EL}(A, \epsilon) = A(\x)+Q_\mathcal{N}w_\mathcal{N}$ where $w_\mathcal{N}\overset{d}{=} \Lap\left(\Delta_1(Q_\mathcal{N}^\top A))/\epsilon\right)^{n-n_c}$ is $(\epsilon, 0)$-\goodprivacy for query $A$ and invariant $C$. 
\end{corollary}

We have thus far discussed four mechanisms, respectively derived using the projection and extension frameworks, and employing Gaussian and Laplace errors. In practice, a data curator would choose either projected mechanisms if seeking to impose invariants on an existing differentially private mechanism, and either extension mechanisms if augmenting a smaller private query while staying compatible with the invariants. The curator would prefer the Laplace mechanisms over the Gaussian ones if a pure (i.e. $\delta = 0$) subspace differential privacy guarantee is sought, although at the expense of heavier-tailed noises which may be undesirable for utility purposes. In what follows, we discuss mechanism options for the curator, if utility considerations are the most salient.


\subsection{Optimal Induced Subspace Differentially Private Mechanisms for Linear Queries}\label{sec:linear}


    
As a consequence of \Cref{thm:connect}, \Cref{cor:err_constrain2dp} translates optimal accuracy enjoyed by a differentially private mechanism to optimal accuracy by an \goodprivacy mechanism.  Let $\opt_{\epsilon, \delta}(A)$ be the optimal error achievable by any $(\epsilon, \delta)$-differentially private mechanism, and $\opt_{\epsilon, \delta}^C(A)$ be the optimal error by any $(\epsilon, \delta)$-\goodprivacy mechanism 
for query $A$ and invariant $C$.

\begin{corollary}\label{cor:err_constrain2dp}  For all $\epsilon, \delta\ge 0$, general query $A:\mathcal{X}^*\to \mathbb{R}^n$, and linear equality invariant $C$, 
$\opt_{\epsilon, \delta}^C(A) = \opt_{\epsilon, \delta}(\Pi_{\mathcal{N}}A)$.  
\end{corollary}
We defer the proof to Appendix~\ref{app:optimal}.  Informally, for any differentially private mechanism we use the extension framework in \Cref{thm:connect} to construct an \goodprivacy mechanism.  Because our proof is constructive, we can translate existing near optimal differentially private mechanisms to \goodprivacy ones.  

We demonstrate this translation with a near-optimal (i.e. mean squared error with a small multiplicative factor) correlated Gaussian mechanism for linear queries from  \citet{nikolov2013geometry}.
Specifically, first design a differential private mechanism $\hat{M}$ for $A_\mathcal{N}:=Q_\mathcal{N}^\top A$, and extend it to a subspace differentially private mechanism by Theorem~\ref{thm:connect}. Because the mean squared error is invariant under rotation, $\err_{\hat{M}}(\Pi_\mathcal{N}A) = \err_{\hat{M}}(A_\mathcal{N})$.  Therefore, if $\hat{M}$ is the near optimal correlated Gaussian noise mechanism for $A_\mathcal{N}$, the resulting \goodprivacy mechanism  is also near optimal by \Cref{cor:err_constrain2dp}.  Details of this mechanism is spelled out in \Cref{alg:subspace_gaussian} which we, together with the proof for \Cref{thm:sdp_corr_gaussian} below, defer to Appendix~\ref{app:optimal}.

\begin{theorem}\label{thm:sdp_corr_gaussian}
Given $\epsilon, \delta>0$, a linear query $A: \mathbb{R}^d\to \mathbb{R}^n$, and a linear equality invariant $C\in \mathbb{R}^{n_c\times n}$, \Cref{alg:subspace_gaussian} is an efficient $(\epsilon, \delta)$-\goodprivacy mechanism such that for all small enough $\epsilon$ and all $\delta$ small enough with respect to $\epsilon$ satisfies 
$$err_{\mathcal{M}}(A) = O(\log^2(n-n_c)\log 1/\delta)\opt_{\epsilon, \delta}^{C}(A).$$
\end{theorem}

We may use the same idea to convert an $k$-norm mechanism~\citep{hardt2010geometry,bhaskara2012unconditional} to an $(\epsilon, 0)$-\goodprivacy one.  \citet{bhaskara2012unconditional} proposed an $(\epsilon, 0)$-differentially private $k$-norm mechanism whose approximation ratio of mean squared error is $O((\log n)^2)$ for any linear query with image in $\mathbb{R}^n$.  We can run the $k$-norm mechanism on query $A_\mathcal{N}$ whose mean squared error is $O((\log (n-n_c))^2)\opt_{\epsilon,0}(A_\mathcal{N})$.  Then, by \Cref{cor:err_constrain2dp}, the output can be converted to an $(\epsilon, 0)$-\goodprivacy mechanism, with an approximation ratio $O((\log n-n_c)^2)$.

\section{Statistical and Practical Considerations}\label{sec:discussion}

\noindent\textbf{Unbiasedness of projection algorithms}
The projection algorithms proposed in this paper, be they Laplace or Gaussian, are provably \emph{unbiased} due to their additive construction. In fact, we have the following result.
\begin{corollary}\label{cor:unbiasedness}
Any mechanism of the form $\mathcal{M}(\x):=A(\x)+\Pi_{\mathcal{N}}\mathbf{e}$ as defined in Corollary~\ref{cor:projection},  where ${\bf e}$ is random noise with $\E\left({\bf e}\right) = 0$, is unbiased in the sense that
\begin{equation*}
\E\left[\mathcal{M}({\bf x}) \mid  A({\bf x})\right] = A\left({\bf x}\right).    
\end{equation*}
\end{corollary}
Corollary~\ref{cor:unbiasedness} stands because the conditional expectation of its noise component, $\E\left[\Pi_{\mathcal{N}}{\bf e} \mid A({\bf x})\right]$, is zero, due to the independence of ${\bf e}$ from $A({\bf x})$, and the nature of the projection operation. All projection mechanisms proposed in this paper are of this type, hence are unbiased. As for the proposed extension algorithms, their purpose is to augment existing DP mechanisms in a way that satisfy mandated invariants, thus their unbiasedness hinge on the unbiasedness of the initial mechanism which they extend. For applications in which the data curator has the freedom to design the privacy mechanism from scratch, projection mechanisms are the recommended way to proceed. Indeed, both numerical demonstrations in Section~\ref{sec:example} applied to the county-level 2020 Census demonstration data and the spatio-temporal university campus data utilize the projection mechanisms, guaranteeing the unbiasedness of the sanitized data products under their respective invariant constraints.

\noindent\textbf{Transparency and statistical intelligibility}
Subspace differentially private additive mechanisms carry a special advantage when examined through the lens of downstream statistical analysis of the output query. All Gaussian and Laplace mechanisms examined in this paper (\cref{cor:projected_gaussian,cor:extended_gaussian,cor:projected_laplace,cor:extended_laplace,thm:sdp_corr_gaussian}), be they obtained via projection or extension, linearly combine the confidential query with a noise term that is publicly specified. Just like standard differential privacy, mechanisms of subspace differential privacy described in this paper are \emph{transparent} \citep{abowd2016economic}, a prized property that brought revolutionary change to the literature of statistical disclosure limitation by ridding obscure procedures. Moreover, the employed noise terms have probability distributions that are fully characterized and independent of the confidential query. This grants the mechanisms \emph{statistical intelligibility} \citep{gong2020congenial}, making the output query eligible for both analytical and simulation-based (such as bootstrap) statistical analysis and uncertainty quantification.


In the special case that the original unconstrained differentially private mechanism is spherical Gaussian, defined in~\Cref{lem:gaussian}, the \goodprivacy mechanism resulting from projection produces a random query that is distributionally equivalent to that obtained via the standard probabilistic conditioning of the unconstrained mechanism, where the conditioning event is precisely the invariants that the curator seeks to impose. 

\begin{theorem}\label{thm:conditioning}
If the additive mechanism $M$ in~\Cref{cor:projection} is spherical Gaussian as defined in~\Cref{lem:gaussian}, the corresponding modified mechanism $\mathcal{M}$ has a probability distribution equivalent to the distribution of $M$ conditional on the invariant being true. That is,
$$
\mathcal{M}(\x)  \overset{d}{=} {M}(\x) \mid C{M}(\x)=CA(\x).$$
\end{theorem}

The proof of \Cref{thm:conditioning} is given in Appendix~\ref{app:spherical-proof}.  The equivalence with conditionalization is particularly valuable for Bayesian inference based on the privatized query, as the analyst may coherently utilize all available information.  We note here however, that \Cref{thm:conditioning} results from unique properties of the spherical Gaussian distribution. In general, the projection operation aligns closer with marginalization, and cannot produce the equivalent distribution as conditionalization. Nevertheless, the happy statistical consequence of \Cref{thm:conditioning} may still be widely impactful, thanks to the ubiquity of the spherical Gaussian mechanism.

\noindent\textbf{Implementation: distributed privatization}
In local differential privacy, we consider the identity mapping as the query function $A$, and the private mechanism directly infuses entry-wise noise into the entire confidential dataset before releasing it to the public. The confidential dataset, $\x$, is often gathered by a number of local agents -- nodes, sensors, survey workers -- each responsible for one (or more) entries of $\x$. Distributed privatization can be valuable in local differential privacy, as it ensures individual data contributors' information is protected the moment it leaves the local agent. 

For all additive subspace differentially private mechanisms proposed in this work, distributed privatization may be achieved,  if the local agents are capable of simulating the same noise component. The synchronized simulation can be implemented --  hardware permitting -- by sharing a common seed across the different local sensors or workers. An instance of distributed privatization is spelled out in Algorithm~\ref{alg:distributed} in Appendix~\ref{app:distributed}, which works for arbitrary linear equality invariant $C$.

\section{ Numerical Examples}\label{sec:example}




\subsection{2020 Census Demonstration Data}

We consider the publication of induced subspace differentially private county-level Census population counts, subject to the invariant of state population size, using the November 2020 vintage privacy-protected demonstration files curated by IPUMS NHGIS \cite{ipums2020das}. These data files link together the original tables from the 2010 Census Summary Files (CSF), here treated as the confidential values, and the trial runs of the Census Bureau's 2020 Disclosure Avoidance System (DAS) applied to the CSF. All these datasets are publicly available at the cited source. For our demonstration, the privacy loss budget is set to accord exactly to the Census Bureau's specification, with $\epsilon = 0.192 = 4 \text{ (total)} \times 0.16 \text{ (county level)} \times 0.3 \text{ (population query)}$. 

Right panel of Figure~\ref{fig:census_trend} showcases the county-level errors from ten runs of the projected Laplace $(\epsilon,0)$-induced subspace differentially private mechanism of \Cref{cor:projected_laplace}, applied to the counties of Illinois arranged in increasing true population sizes.  Compared with the DAS errors (red squares) which show a clear negative bias trend, the proposed mechanism is provably unbiased, due to its additive errors being projected from unbiased and unconstrained random variables. On the other hand, these errors span a similar scale compared to the DAS errors. Figure~\ref{fig:census_sdp_county2} in Appendix~\ref{app:example} shows the application of the projected Laplace $(\epsilon,0)$-induced subspace differentially private mechanism to an additional ten states, for which the TopDown algorithm incurred decidedly negatively biased errors. Details of how these states were identified are given in Section~\ref{sec:intro}. We can make similar observations from Figure~\ref{fig:census_sdp_county2} about the errors from the proposed mechanism as we did from Figure~\ref{fig:census_trend}, including their unbiasedness yet a similar error scale compared to the DAS errors.

\subsection{Spatio-temporal Dataset}

We consider the publication of time series derived from WiFi log data on connections of mobile devices with nearby access points from a large university campus (Tsinghua University)~\cite{sui2016characterizing} consisting of $3243$ fully anonymized individuals and the top $20$ most popularly visited buildings in one day. \footnote{Data was collected under the standard consent for Wifi access on university campus. Interested reader may contact the authors of \cite{sui2016characterizing} to inquire access to the dataset.} The raw data recorded whether an individual appears in each of the building in each of the hours on one day. The data were tabulated into hourly time series for $14$ clusters of individuals obtained through simple $K$-means, to represent hypothetical group memberships with distinct travel patterns. 

The invariants we consider are of two types, motivated by needs for building management, energy-control, and group activity scheduling: 1) the total number of person-hours spent at each building every hour from all groups, and 2) the total number of person-hours spent at each building by every group over 24 hours. The query under consideration is $14\text{ (groups)}\times 24\text{ (hours)}\times 20\text{ (location)} = 6720$ dimensional, subject to a $(24 + 14 - 1)\times 20 = 740$-dimensional linear constraint. 
	
We apply the projection Gaussian mechanism in \Cref{cor:projected_gaussian}.
The comparison of confidential data and one run of the induced subspace differentially private mechanism is displayed in Figure~\ref{fig:tsinghua}. 
The mechanism is again provably unbiased, although the errors exhibit a slight loss of scale due to the numerous linear constraints imposed. Over 50 repetitions of the mechanism, the median standard deviation of the elementwise additive errors is $0.88$ (relative to one unit), with $5\%$ and $95\%$ quantiles at $(0.86, 0.91)$ respectively. Results of the simulation are displayed in Figure~\ref{fig:tsinghua} of Appendix~\ref{app:example}.

\section{Conclusion and Future Work} \label{sec:conclusion}

In this paper, we proposed the concept of subspace differential privacy to make explicit the mandated invariants imposed on private data products, and discussed the projection and extension designs of induced differentially private mechanisms. The invariants we consider are in the form of linear equalities, including sums and contingency table margins as often encountered in applications including the U.S. Decennial Census and spatio-temporal datasets.

An important type of invariants not addressed in this paper are inequalities, such as nonnegativity and relational constraints (e.g. the population size must be larger or equal to the number of households in a geographic area). However, we note that an important premise to the unbiasedness achieved by subspace differentially private mechanisms, as discussed in Section~\ref{sec:discussion}, is that the mechanism admits equality invariants only. If inequality invariants must be imposed, unbiased privacy mechanisms can be inherently difficult to design. As \citet{DBLP:journals/corr/abs-2010-04327} discussed, the bias induced by projection-type post-processing of noisy measurements is attributable to the non-negativity constraints imposed on them. This raises the question of the appropriateness of inequality invariants on the sanitized output, if unbiasedness is simultaneously required.

Also not considered are invariants for binary and categorical attributes, taking values in a discrete space. These invariants differ from real-valued linear equality invariants, because in general they cannot be realized by an additive mechanism with a noise term independent of the confidential data value. While the notion of subspace differential privacy can be extended to these cases, the design of accompanying privacy mechanisms that also enjoy good statistical and implementation properties remains a subject of future research.




\bibliographystyle{plainnat}
\bibliography{main,jiepub}
\newpage

\appendix





\section{Basic properties of subspace and induced subspace differential privacy}\label{app:subspace}


\begin{proof}[Proof of \Cref{prop:nested}]
Because $\mathcal{V}_2\subseteq \mathcal{V}_1$, for any $v_1\in \mathcal{V}_1$, there exist unique $v_2\in \mathcal{V}_2$ and $w\in \mathcal{V}_2^\perp\cap \mathcal{V}_1$ such that $v_1 = v_2+w$.  We set $W = \mathcal{V}_2^\perp\cap \mathcal{V}_1$ be orthogonal to $\mathcal{V}_2$ and in $\mathcal{V}_1$.  For any set $S\subseteq \mathbb{R}^n$, we define $S+W:=\{v = s+w: s\in S, w\in W\}$ which is also measurable when $S$ is measurable, and $\Pi_{\mathcal{V}_1}^{-1}(S):=\{v:\Pi_{\mathcal{V}_1} v\in S\}$.  By direct computation, for any $S$,
\begin{equation}\label{eq:nested1}
    \Pi_{\mathcal{V}_2}^{-1}(S) 
    = \Pi_{\mathcal{V}_1}^{-1}(S+W).
\end{equation}  Therefore, for any measurable set $S$, neighboring database $\x\sim \x'$, and $\mathcal{V}_1$-subspace differentially private mechanism $M$, we have
\begin{align*}
    &\Pr[\Pi_{\mathcal{V}_2}M(\x)\in S]\\
    =& \Pr[\Pi_{\mathcal{V}_1}M(\x)\in S+W]\tag{by \Cref{eq:nested1}}\\
    \le& e^\epsilon \Pr[\Pi_{\mathcal{V}_1}M(\x')\in S+W]+\delta \tag{$\mathcal{V}_1$-subspace dp}\\
    =& e^\epsilon \Pr[\Pi_{\mathcal{V}_2}M(\x')\in S]+\delta \tag{by \Cref{eq:nested1}}
\end{align*}
This completes the proof.
\end{proof}

\subsection{Composition and post processing}


\begin{proof}[Proof of \cref{prop:composition}]
Let the null space of $C_{1,2}$ be $N_{1,2}:=\{(v_1, v_2)\in Y_1\times Y_2: v_1\in N_1, v_2\in N_2\}$ where $N_1$ and $N_2$ are the null space of $C_1$ and $C_2$ respectively.  Then for all neighboring databases $x$ and $x'$ and an outcome $(y_1, y_2)\in Y_1\times Y_2$ we have
\begin{align*}
    &\Pr[\Pi_{N_{1,2}}(M_1,M_2)(x) (y_1, y_2)]\\
    =& \Pr[\Pi_{N_1}(M_1(x)) = y_1]\Pr[ 
\Pi_{N_2}(M_2(x)) = y_2]\\
\le& e^{\epsilon_1} \Pr[\Pi_{N_1}(M_1(x')) = y_1]\cdot e^{\epsilon_2} \Pr[\Pi_{N_2}(M_2(x')) = y_2]\\
=& \exp(\epsilon_1+\epsilon_2)\Pr[\Pi_{N_{1,2}}(M_1,M_2)(x') = (y_1, y_2)].
\end{align*}
Finally, the invariant also holds, because $C_{1,2}(M_1(x),M_2(x))$ equals $(C_1A_1(x), C_2A_2(x))$.
\end{proof}
\section{Correlated Gaussian Mechanism by {\citet{nikolov2013geometry}}{Nikolov et al.}}\label{sec:pre_corr}
For completeness, in this section, we state the correlated Gaussian noise mechanism (\Cref{alg:dense_gaussian}) and one main theorem (\Cref{thm:dp_correlatedgaussian}) in \citet{nikolov2013geometry}.  We are going to modify the mechanism (\Cref{alg:dense_gaussian}) to an \goodprivacy one (\Cref{alg:subspace_gaussian}) in \Cref{sec:linear}.  

\begin{algorithm}[htb]
\caption{Correlated Gaussian Noise mechanism $\mathcal{M}_{CG}$~\cite{nikolov2013geometry}}
\label{alg:dense_gaussian}
\textbf{Input}: $(A,\mathbf{h}, \epsilon, \delta)$ where linear query $A = (a_i)^d_{i = 1}\in \mathbb{R}^{n\times d}$ has full rank $n$, the histogram of a database $\mathbf{h}\in \mathbb{R}^d$, privacy constrains $\epsilon, \delta$
\begin{algorithmic}[1] 
\STATE Let $c_{\epsilon, \delta}:=\epsilon^{-1}(1+\sqrt{1+\ln (1/\delta)})$.
\STATE Compute $E = FB^n_2$, the minimum volume enclosing ellipsoid of $K = AB_1$ where $B_2^n = \{\mathbf{x}\in \mathbb{R}^n:\|\mathbf{x}\|_2 = 1\}$ and $B_1 = \{\mathbf{x}\in \mathbb{R}^d:\|\mathbf{x}\|_1 = 1\}$;
\STATE Let $u_i$ $i = 1, \dots, n$ be the left singular vectors of $F$ corresponding to singular values $\sigma_1\ge\dots\ge \sigma_n$;
\IF{$d = 1$}
\STATE \textbf{return} $U_1 = u_1$.
\ELSE
\STATE Let $U_1 = (u_i)_{i> n/2}$ and $V = (u_i)_{i\le n/2}$;
\STATE Recursively compute a base decomposition $V_2, \dots, V_k$ of $V^\top A$;
\STATE Let $U_i = VV_i$ for each $1<i\le k$ where $k = \lceil 1+\log n\rceil$.
\ENDIF
\FOR{$i = 1, \ldots, k$}
\STATE let $r_i = \max_{j = 1}^d\|U_i^\top a_j\|_2$. 
\STATE Sample $\mathbf{w}_i\sim N(0; c_{\epsilon, \delta})^{n_i}$
\ENDFOR
\STATE \textbf{return} $A\mathbf{h}+\sqrt{k}\sum^k_{i = 1}r_i U_i \mathbf{w}_i$
\end{algorithmic}
\end{algorithm}

\begin{theorem}[Theorem 13 in \citet{nikolov2013geometry}]\label{thm:dp_correlatedgaussian}
Algorithm~\ref{alg:dense_gaussian} $\mathcal{M}_{CG}$, is $(\epsilon, \delta)$-differentially private and for all small enough $\epsilon$ and satisfies 
$$err_{\mathcal{M}_{CG}}(A) = O(\log^2n\log 1/\delta)\opt_{\epsilon, \delta}(A)$$
for all $\delta$ small enough with respect to $\epsilon$.
\end{theorem}

\section{Details and Proofs in \Cref{sec:design}}\label{app:general}
\subsection{Proofs in the Extension Framework}
\connect*
\begin{proof}[Proof of \Cref{thm:connect}]
First, by the definition of $\hat{M}$, we have
\begin{equation}\label{eq:connect1}
    \Pi_{\mathcal{N}}\hat{M}(\mathbf{x})= \Pi_{\mathcal{N}}\left(\mathcal{M}(\mathbf{x})-\Pi_{\mathcal{R}} A(\mathbf{x})\right) = \Pi_{\mathcal{N}}\mathcal{M}(\mathbf{x}).
\end{equation}
If $\mathcal{M}$ is \goodprivacy, $C\hat{M}(\mathbf{x}) = C\mathcal{M}(\mathbf{x})-C\Pi_{\mathcal{R}} A(\mathbf{x}) = 0$, for all $\x$, since $C$ has full rank, the image of $\hat{M}$ is in $\mathcal{N}$.
Because the image of $\hat{M}$ is in $\mathcal{N}$, by \Cref{eq:connect1}, for any measurable set $S$, and neighboring databases $\mathbf{x}, \mathbf{x}'$, $\Pr[\hat{M}(\mathbf{x})\in S] = \Pr[\Pi_\mathcal{N}\hat{M}(\mathbf{x})\in \Pi_\mathcal{N}S] = \Pr[\Pi_{\mathcal{N}}\left(\mathcal{M}(\mathbf{x})\right)\in \Pi_\mathcal{N}S]$ and $\Pr[\hat{M}(\mathbf{x}')\in S] = \Pr[\Pi_{\mathcal{N}}\left(\mathcal{M}(\mathbf{x}')\right)\in \Pi_\mathcal{N}S]$.  Additionally, $\Pr[\Pi_{\mathcal{N}}\left(\mathcal{M}(\mathbf{x})\right)\in \Pi_\mathcal{N}S]\le e^\epsilon\Pr[\Pi_{\mathcal{N}}\left(\mathcal{M}(\mathbf{x}')\right)\in \Pi_\mathcal{N}S]+\delta$ because $\mathcal{M}$ is \goodprivacy and thus $\mathcal{N}$-subspace differentially private.  With these two,  $\Pr[\hat{M}(\mathbf{x})\in S] = e^\epsilon\Pr[\hat{M}(\mathbf{x}')\in S]+\delta$ so $\hat{M}$ is $(\epsilon, \delta)$-differentially private.

Conversely, suppose $\hat{M}$ is $(\epsilon, \delta)$-differentially private with image contained in $\mathcal{N}$.  By \Cref{eq:connect1}, 
$\Pi_{\mathcal{N}}\left(\mathcal{M}(\mathbf{x})\right) = \hat{M}(\mathbf{x})$ which has identical probability distribution as $\Pi_{\mathcal{N}}M\left(\x\right)$.  Thus, 
$\mathcal{M}$ is $\mathcal{N}$-subspace $\left(\epsilon,\delta\right)$-differentially private.  
For linear equality constraint, because $\hat{M}(\mathbf{x})\in \mathcal{N}$, $C\left(\mathcal{M}(\mathbf{x})\right) = C\Pi_{\mathcal{R}} A(\mathbf{x}) = CA(\mathbf{x})$.
\end{proof}

\subsection{Proofs in Mechanisms for General Queries}
By \Cref{cor:projection} and \Cref{thm:connect}, the mechanisms in \Cref{cor:projected_gaussian,cor:extended_gaussian,cor:projected_laplace,cor:extended_laplace} are \goodprivacy.  Thus, it remains to show the error bounds in \Cref{cor:projected_gaussian,cor:extended_gaussian}.

\begin{proof}[Proof of \Cref{cor:projected_gaussian}] Because $\mathbf{e}_G$ is a unbiased Gaussian with covariance $(c_{\epsilon, \delta}\Delta_2(A))^2\mathbf{I}_n$ where $\mathbb{I}_n$ is the identity matrix of dimension $n$, after projection, $\Pi_\mathcal{N}\mathbf{e}_G$ is a unbiased Gaussian with covariance $(c_{\epsilon, \delta}\Delta_2(A))^2\Pi_\mathcal{N}^\top \mathbf{I}_n\Pi_\mathcal{N}$.  Thus, the mean squared error is 
\begin{align*}
    &\E[\|A(\x)-\mathcal{M}_{PG}(\x)\|^2_2]\\
    =& \E[\|\Pi_\mathcal{N}\mathbf{e}_G\|^2_2]\\
    =& tr((c_{\epsilon, \delta}\Delta_2(A))^2\Pi_\mathcal{N}^\top \mathbf{I}_n\Pi_\mathcal{N})\\
    =& (c_{\epsilon, \delta}\Delta_2(A))^2 (n-n_c)
\end{align*}
where $tr(B)$ is the trace of a symmetric matrix $B$.
\end{proof} 

\begin{proof}[Proof of \Cref{cor:extended_gaussian}] Because $\mathbf{e}_{\mathcal{N}}$ is a unbiased Gaussian with covariance $(c_{\epsilon, \delta}\Delta_2(Q^\top_\mathcal{N} A))^2\mathbf{I}_{n-n_c}$, the means squared error is \begin{align*}
    &\E[\|A(\x)-\mathcal{M}_{EG}(\x)\|^2_2]\\
    =& (c_{\epsilon, \delta}\Delta_2(A))^2tr(Q_\mathcal{N}^\top \mathbf{I}_{n-n_c}Q_\mathcal{N})\\
    =& (c_{\epsilon, \delta}\Delta_2(Q^\top_\mathcal{N} A))^2 (n-n_c).
\end{align*}
\end{proof}

Finally, note that 
\begin{align*}
    \Delta_2(Q^\top_\mathcal{N}A) =& \sup_{\x\sim \x'}\|Q^\top_\mathcal{N}A(\x)-Q^\top_\mathcal{N}A(\x')\|_2\\
    =&\sup_{\x\sim \x'}\|Q^\top_\mathcal{N}\left(A(\x)-A(\x')\right)\|_2\\
    \le& \sup_{\x\sim \x'}\|A(\x)-A(\x')\|_2 = \Delta_2(A),
\end{align*} so the mean squared error of the extended Gaussian mechanism in \Cref{cor:extended_gaussian} is always less than or equal to the error of the projected Gaussian mechanism in \Cref{cor:projected_gaussian}.  Intuitively, the noise of the extended Gaussian mechanism only depends on the sensitivity of $A$ in the null space, $Q_\mathcal{N}^\top A$.  However, the projected Gaussian mechanisms modify an already differentially private mechanism to \goodprivacy, and it may introduce additional noise if $A$ is very sensitive in invariant space $\mathcal{R}$.

\subsection{Proofs in Mechanisms for Linear Queries}\label{app:optimal}

\begin{proof}[Proof of \Cref{cor:err_constrain2dp}]
For any $\mathcal{N}$-subspace differentially private mechanism $\mathcal{M}$, let $\hat{M} := \Pi_{\mathcal{N}}\mathcal{M}$.  Then the squared error can be decomposed as
\begin{equation}\label{eq:err1}
    \|\mathcal{M}(\mathbf{x})-A(\mathbf{x})\|^2_2 =    \|\hat{M}(\mathbf{x})-\Pi_{\mathcal{N}}A(\mathbf{x})\|^2_2+\|\Pi_{\mathcal{R}}\mathcal{M}(\mathbf{x})-\Pi_{\mathcal{R}}A(\mathbf{x})\|^2_2.
\end{equation}
By \Cref{thm:connect}, $\hat{M}$ is differentially private, so
$\mathbb{E}\left[\|\hat{M}(\mathbf{x})-\Pi_{\mathcal{N}}A(\mathbf{x})\|^2_2\right]\ge \opt_{\epsilon, \delta}(\Pi_{\mathcal{N}}A)$.  Therefore
$\opt_{\epsilon, \delta}^C(A) \ge \opt_{\epsilon, \delta}(\Pi_{\mathcal{N}}A)$ by \Cref{eq:err1}.

Conversely, for all differentially private mechanism $\hat{M}$, we define 
$\mathcal{M}(\mathbf{x}) := \Pi_{\mathcal{N}}\hat{M}(\mathbf{x})+\Pi_{\mathcal{R}}A(\mathbf{x})$. By post processing property $\Pi_{\mathcal{N}}\hat{M}$ is differentially private, so $\mathcal{M}$ is $\mathcal{N}$-subspace differentially private by Theorem~\ref{thm:connect}.  By \Cref{eq:err1}, we have
\begin{align*}
    &\mathbb{E}\left[\|\mathcal{M}(\mathbf{x})-A(\mathbf{x})\|^2_2\right]\\
    =& \mathbb{E}\left[ \|\Pi_\mathcal{N}\hat{M}(\mathbf{x})-\Pi_{\mathcal{N}}A(\mathbf{x})\|^2_2+\|\Pi_{\mathcal{R}}\mathcal{M}(\mathbf{x})-\Pi_{\mathcal{R}}A(\mathbf{x})\|^2_2\right]\\
    =&\mathbb{E}\left[\|\Pi_\mathcal{N}\hat{M}(\mathbf{x})-\Pi_{\mathcal{N}}A(\mathbf{x})\|^2_2\right]\\
    \le& \mathbb{E}\left[\|\hat{M}(\mathbf{x})-\Pi_{\mathcal{N}}A(\mathbf{x})\|^2_2\right]
\end{align*}
Therefore, $\opt_{\epsilon, \delta}^C(A) \le \opt_{\epsilon, \delta}(\Pi_{\mathcal{N}}A)$.
\end{proof}

\begin{algorithm}[htb]
\caption{Subspace Gaussian Noise mechanism}
\label{alg:subspace_gaussian}
\textbf{Input}: $(A,C, \mathbf{h},\epsilon, \delta)$ where  $A:\mathbb{R}^d\to \mathbb{R}^n$ is a linear query function with full rank, in a linear equality invariant $C:\mathbb{R}^n\to \mathbb{R}^{n_c}$ with row space $\mathcal{R}$ and null space $\mathcal{N}$, the histogram of a database $\mathbf{h} = \hist(\x)\in \mathbb{R}^d$, and privacy constrains $\epsilon, \delta$.
\begin{algorithmic}[1] 
\STATE Compute $Q_\mathcal{N}\in \mathbb{R}^{n\times (n-n_c)}$ an collection of an orthonormal basis of $\mathcal{N}$,
\STATE Let $c_{\epsilon, \delta}:=\epsilon^{-1}(1+\sqrt{1+\ln (1/\delta)})$ and $A_\mathcal{N} := Q_\mathcal{N}^\top A$ \COMMENT{full rank and $Q_\mathcal{N}A_\mathcal{N} = \Pi_\mathcal{N} A$}
\STATE Compute $E = FB^n_2$, the minimum volume enclosing ellipsoid of $K = A_\mathcal{N}B_1$ where $B_2^n = \{\mathbf{x}\in \mathbb{R}^n:\|\mathbf{x}\|_2 = 1\}$ and $B_1 = \{\mathbf{x}\in \mathbb{R}^d:\|\mathbf{x}\|_1 = 1\}$;
\STATE Let $u_i$ $i = 1, \dots, n$ be the left singular vectors of $F$ corresponding to singular values $\sigma_1\ge\dots\ge \sigma_n$;
\IF[Decompose $A_\mathcal{N}$ in to spaces $U_i$ $i = 1,\dots,k$]{$d = 1$}
\STATE $U_1 = u_1$.
\ELSE
\STATE Let $U_1 = (u_i)_{i> n/2}$ and $V = (u_i)_{i\le n/2}$;
\STATE Recursively compute a base decomposition $V_2, \dots, V_k$ of $V^\top A_\mathcal{N}$;
\STATE Let $U_i := VV_i$ with dimension $n_i$ for each $1<i\le k$ where $k = \lceil 1+\log n\rceil$;
\ENDIF
\FOR[Compute the noise in each space $U_i$]{$i = 1, \ldots, k$}
\STATE let $r_i = \max_{j = 1}^d\|U_i^\top a_j\|_2$. 
\STATE Sample $\mathbf{w}_i\sim N(0; c_{\epsilon, \delta})^{n_i}$
\ENDFOR
\STATE $\mathbf{z} = A_\mathcal{N}\mathbf{h}+\sqrt{k}\sum^k_{i = 1}r_i U_i \mathbf{w}_i$ which is in $\mathbb{R}^{n-n_c}$.
\STATE \textbf{return} $Q_\mathcal{N}\mathbf{z}+\Pi_{\mathcal{R}}A\mathbf{h}$.
\end{algorithmic}
\end{algorithm}

\begin{proof}[Proof of \Cref{thm:sdp_corr_gaussian}]

We first show the privacy guarantees of our mechanism, $\mathcal{M}$.  By \Cref{thm:dp_correlatedgaussian}, the output of \Cref{alg:dense_gaussian} is $(\epsilon, \delta)$-differentially private, so in  \Cref{alg:subspace_gaussian} $Q_\mathcal{N}\mathbf{z}$ is $(\epsilon, \delta)$ differentially private, and  $Q_\mathcal{N}\mathbf{z}\in \mathcal{N}$.  Thus, by \Cref{thm:connect}, outputting $Q_\mathcal{N}\mathbf{z}+\Pi_{\mathcal{R}}A\mathbf{h}$ is $\mathcal{N}$-subspace $(\epsilon, \delta)$ differentially private.  Finally, the output satisfies linear equality constraints $C$, because $ C(Q_\mathcal{N}\mathbf{z}+\Pi_{\mathcal{R}}A\mathbf{h})= C\Pi_{\mathcal{R}}A\mathbf{h} =  CA\mathbf{h}$.

Now we study the accuracy of our mechanism.  Because
\begin{align*}
    &\|Q_\mathcal{N}\mathbf{z}+\Pi_{\mathcal{R}}A\mathbf{h}-A\mathbf{h}\|^2\\ =&\|Q_\mathcal{N}\mathbf{z}-\Pi_\mathcal{N} A \mathbf{h}\|^2\\
    =& \|Q_\mathcal{N}\left(\mathbf{z}-Q_\mathcal{N}^\top A \mathbf{h}\right)\|^2\\
    =& \|\mathbf{z}-Q_\mathcal{N}^\top A \mathbf{h}\|^2
\end{align*}
the error of our mechanism is equal to $\err_{\mathcal{M}_{CG}}(Q_\mathcal{N}^\top A)$ which the error between the output of \Cref{alg:dense_gaussian} and linear query $Q_\mathcal{N}^\top A$.
  By \Cref{thm:dp_correlatedgaussian}, 
\begin{equation}\label{eq:opt1}
    \err_{\mathcal{M}_g}(Q_\mathcal{N}^\top A ) = O(\log^2(n-n_c)\log 1/\delta)\opt_{\epsilon, \delta}(Q_\mathcal{N}^\top A).
\end{equation}

Now we want to show the optimal error for query $Q_\mathcal{N}^\top A$ is no more than the optimal one for query $\Pi_\mathcal{N}A$.  Formally,
\begin{equation}\label{eq:opt2}
    \opt_{\epsilon, \delta}(Q_\mathcal{N}^\top A)\le \opt_{\epsilon, \delta}(\Pi_\mathcal{N}A).
\end{equation}

For any $(\epsilon, \delta)$-differentially private mechanism $M:\mathbb{R}^d\to \mathbb{R}^{n}$, we have 
\begin{align*}
    \err_{M}(\Pi_\mathcal{N} A) =& \sup_\mathbf{h}\mathbb{E}\left[\|M(\mathbf{h})-\Pi_\mathcal{N} A\mathbf{h}\|^2\right]\\
    \le& \sup_\mathbf{h}\mathbb{E}\left[\|Q_\mathcal{N}^\top (M(\mathbf{h})-\Pi_\mathcal{N} A\mathbf{h})\|^2\right]\tag{$Q_\mathcal{N}$ consists of orthonormal columns}\\
    =& \sup_\mathbf{h}\mathbb{E}\left[\|Q_\mathcal{N}^\top M(\mathbf{h})-Q_\mathcal{N}^\top A\mathbf{h}\|^2\right]\\
    =& \err_{Q_\mathcal{N}^\top M}(Q_\mathcal{N}^\top A).
\end{align*}
Therefore we have a new $(\epsilon, \delta)$-differentially private mechanism $Q_\mathcal{N}^\top M$ so that the error for query $Q_\mathcal{N}^\top A$ is less than or equal to the error of $M$ for query $\Pi_\mathcal{N} A$ which completes the proof of \Cref{eq:opt2}.

Finally by \Cref{cor:err_constrain2dp}, 
\begin{equation}\label{eq:opt3}
     \opt_{\epsilon, \delta}(\Pi_\mathcal{N}A) = \opt_{\epsilon, \delta}^\mathcal{N}(A).
\end{equation}
Combining \Cref{eq:opt1,eq:opt2,eq:opt3} completes the proof.
\end{proof}

\section{Proof of Theorem~\ref{thm:conditioning}}\label{app:spherical-proof}

\Cref{thm:conditioning} is established by recognizing that the multivariate Gaussian is a location-scale family completely characterized by its mean vector and covariance matrix. In $\mathcal{M}$, the projected additive noise $\Pi_\mathcal{N} \mathbf{e}$ has zero mean and covariance matrix $(\Delta_2(A)c_{\epsilon, \delta})^2\Pi_\mathcal{N}$. Hence, $\mathcal{M}$ as a mechanism is unbiased, and has mean $CA(\x)$ and the same covariance matrix. On the other hand, the conditional distribution of the unrestricted mechanism $M$ has mean $CA(\x)$ and covariance matrix $(\Delta_2(A)c_{\epsilon, \delta})^2(\mathbf{I} - C^{\top}(CC^{\top})^{-1}C)$, which is the same as $(\Delta_2(A)c_{\epsilon, \delta})^2\Pi_\mathcal{N}$ with $\Pi_\mathcal{N}$ being the unique orthogonal projection matrix. Since the Gaussian family is closed with respect to linear marginalization and linear conditionalization, we have that the distribution of $\mathcal{M}$  is indeed identical to the conditional distribution of $M$ given $C{M}(\x)=CA(\x)$. 

\section{Distributed mechanism algorithm}\label{app:distributed}
Here we specify our model of distributed computation.   Let $\mathcal{K} = (K_1,\dots,K_m)$ be a partition of $\left\{ 1,\ldots,n\right\}$, with size $m$. That is, each element $K\in\mathcal{K}$ is a subset of observations that one particular agent is responsible for collecting, and there are $m$ agents in total. Without lose of generality, through relabelling we can have $K_\ell = \{|K_{\le \ell}|+1,\dots,|K_{\le \ell}|+|K_\ell|\}\subseteq \left\{ 1,\ldots,n\right\}$ that contains entries with index from $|K_{\le \ell}|+1$ to $|K_{\le \ell}|+|K_\ell|$ where $K_{\le \ell}:=\cup_{\iota<\ell}K_\iota$.  Now we can project the dataset $\x$ into those $m$ partitions.  Let $\Gamma_{\ell}$ denote a $|K_\ell|\times n$ matrix where $(\Gamma_k)_{i,j} = \mathbf{1}\left[j = i+|K_{\le \ell}|\right]$, and ${\bf x}_{\ell}=\Gamma_\ell \x$, be the subvector of ${\bf x}$ whose indices are atoms of $k$.

The following meta-algorithm shows how to adapt any additive privatization mechanisms defined in \Cref{eq:additive} to a distributed one.  Every agent carries out the privatization locally, while ensuring that the aggregated privatized data ${\bf y}$ satisfies the global invariant constraint in \Cref{def:invariant}.
\begin{algorithm}[htb]
\caption{Distributed privatization framework}
\label{alg:distributed}
\textbf{Input}: a database $\x\in \mathbb{R}^N$,  the identity query $A:\mathbb{R}^N\to \mathbb{R}^N$, linear equality invariant $C\in \mathbb{R}^{n_c\times N}$ with rank $n_c$, $\epsilon\in (0,1)$, and $\delta\in (0,1)$;

\textbf{Parameter}: common~seed~$a$,~partition~$\mathcal{K}$, and an additive privatization mechanisms $\mathcal{M}(A,C,\x,\epsilon,\delta)$;

\begin{algorithmic}[1] 
\FOR[in parallel]{agent $K_\ell\in \mathcal{K}$}
\STATE Observe ${\bf x}_{\ell}$;
\STATE Simulate noise $\mathbf{e}$ with common seed $a$ subject to privacy budget $\epsilon, \delta$ and invariant constraint $C$.\footnotemark
\STATE Compute $\mathbf{y}_\ell = \x_\ell+\Gamma_\ell \mathbf{e}$.
\ENDFOR
\STATE \textbf{return} concatenated $(\mathbf{y}_1, \mathbf{y}_2,\dots,\mathbf{y}_m)$.
\end{algorithmic}
\end{algorithm}
\footnotetext{Formally, agent $\ell$ first augments the observation to $\bar{\x}_\ell\in \mathbb{R}^N$ by filling zero to the unknown entries of $\x_\ell$.  Then agent $\ell$ runs $\mathcal{M}$ on input $(A,C,\bar{\x}_\ell,\epsilon, \delta)$, and set $\mathbf{e} = \mathcal{M}(A,C,\bar{\x}_\ell,\epsilon, \delta)-\bar{\x}_\ell$.  Since $\mathcal{M}$ is additive, the noise term is independent of $\bar{\x}_\ell$, and every agent will get the same $\mathbf{e}$.}

When invariants are in place, the feasibility of distributed privatization may be counter-intuitive, because invariants typically induce dependence among entries of the private query. The key here is that for additive mechanisms that impart linear equality invariants, the noise that each local agent infuses does not depend on the confidential $\x$ at all. With a shared seed, every agent can generate the same noise vector ($\mathbf{e}$ in Algorithm~\ref{alg:distributed}), and carry out their privatization task in isolation from others. 

\section{Supplement numerical analysis}\label{app:example}

This appendix displays figures that accompany the numerical analyses of Section~\ref{sec:example}. Details are given therein.




\begin{figure*}
\begin{center}
\includegraphics[width = .47\textwidth]{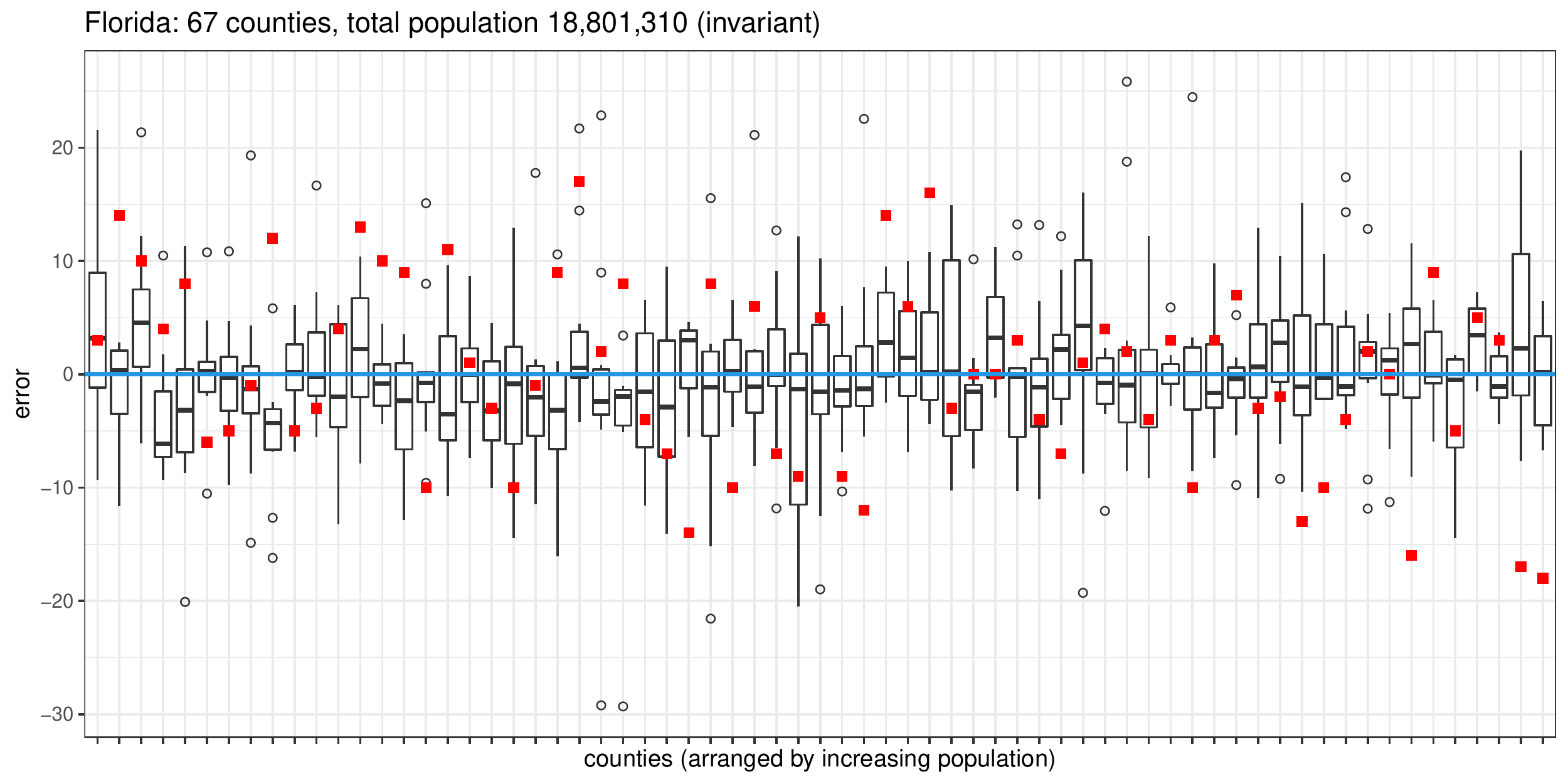}
\includegraphics[width = .47\textwidth]{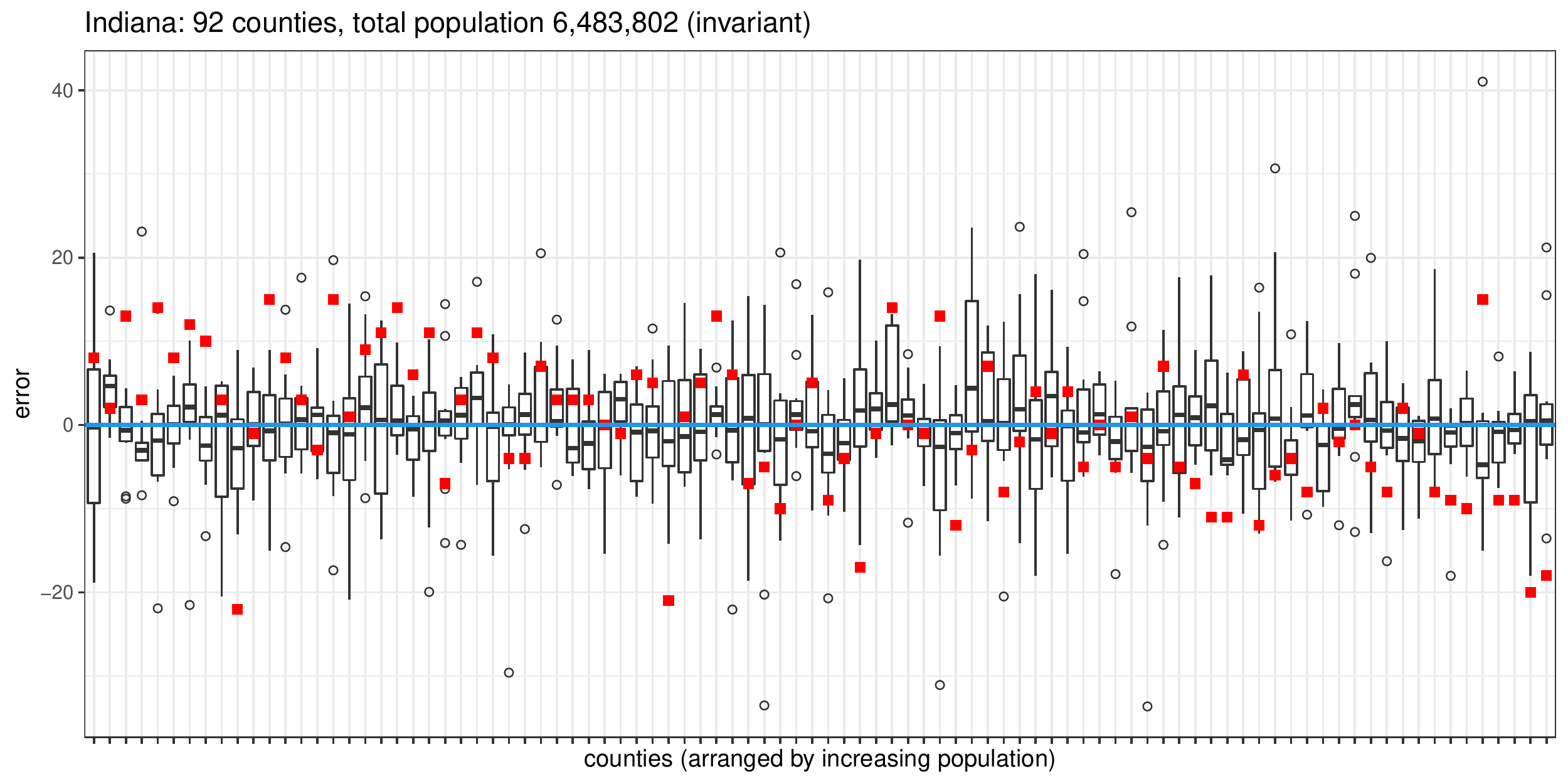}

\includegraphics[width = .47\textwidth]{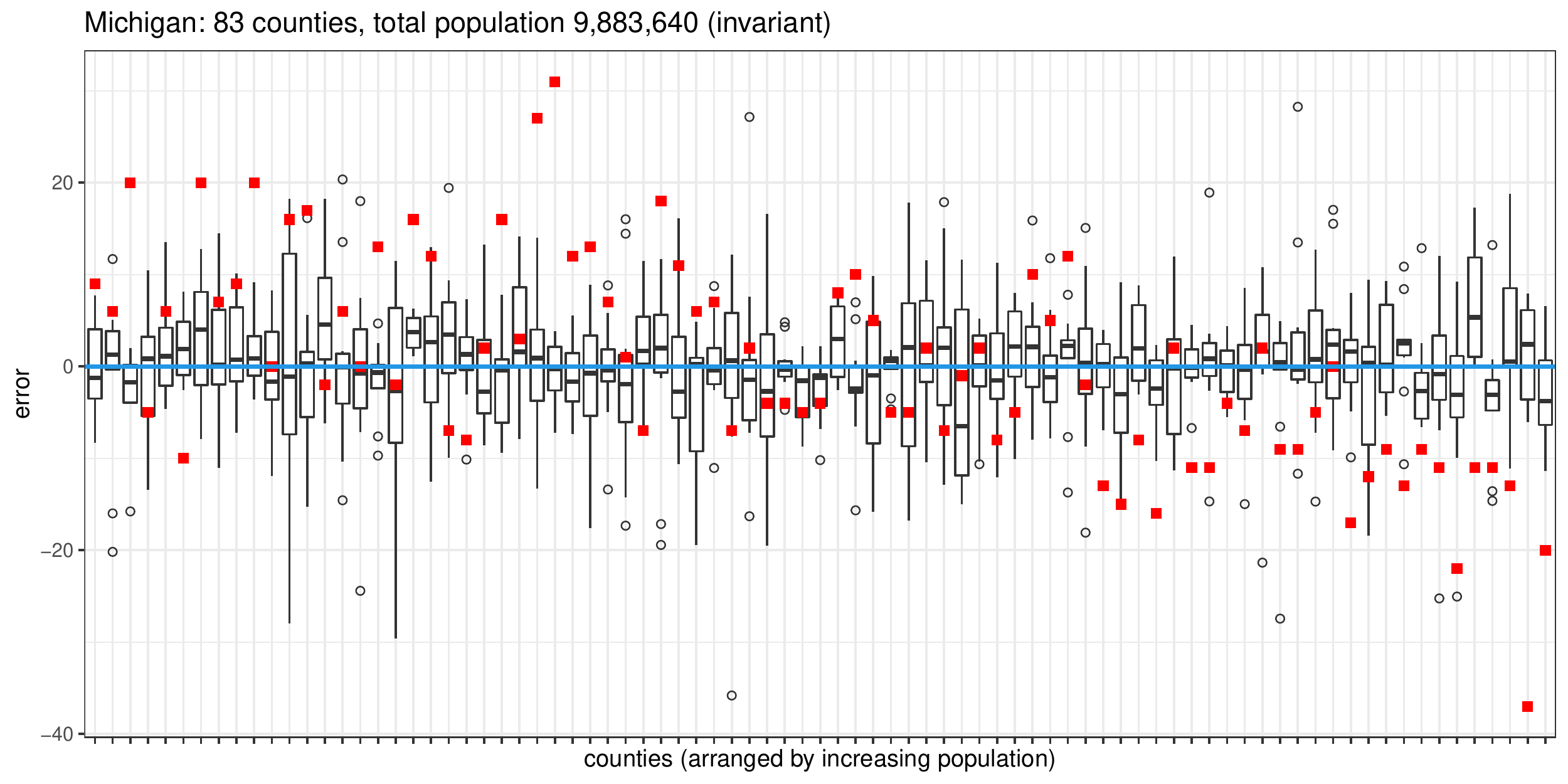}
\includegraphics[width = .47\textwidth]{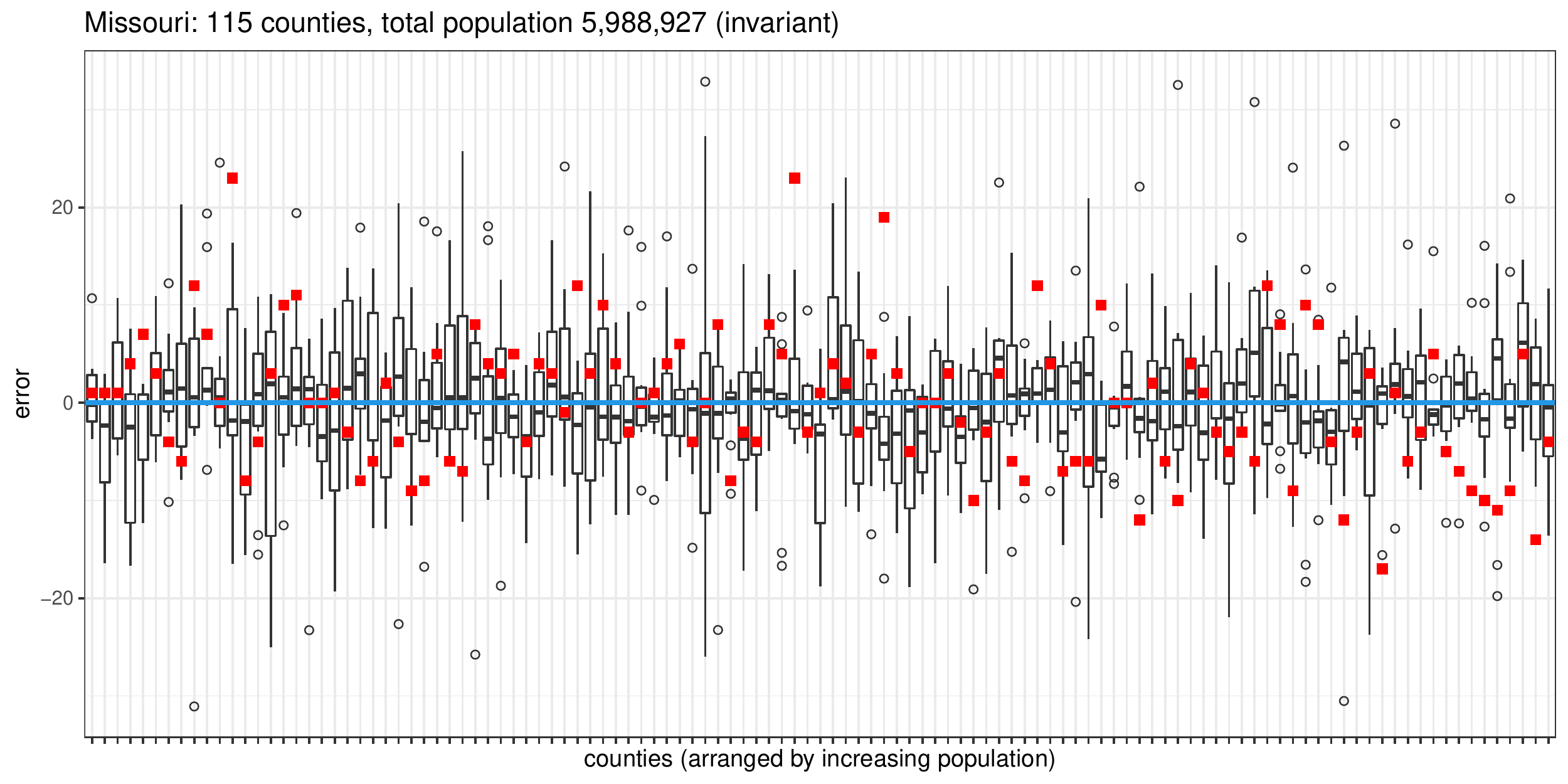}

\includegraphics[width = .47\textwidth]{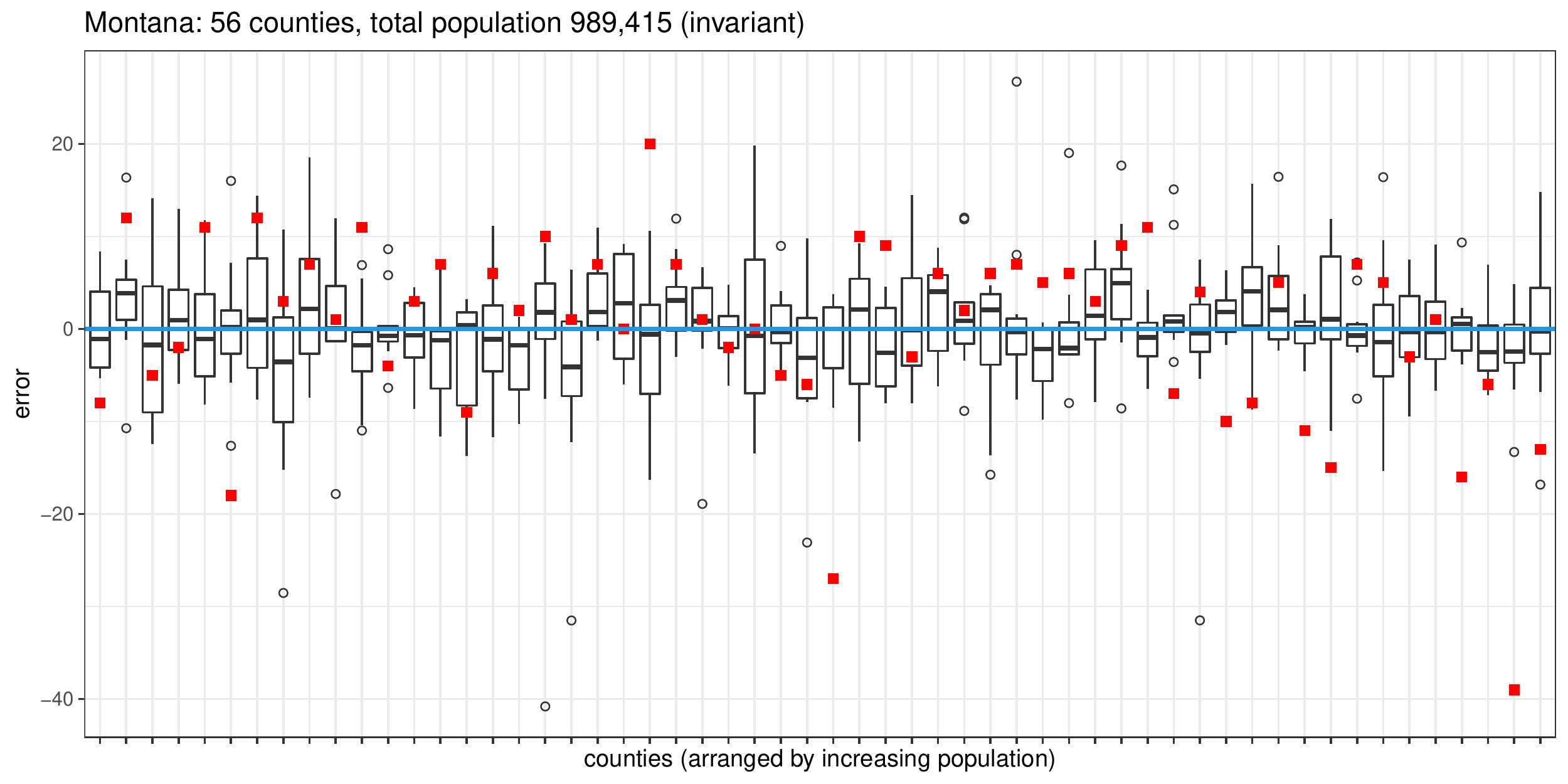}
\includegraphics[width = .47\textwidth]{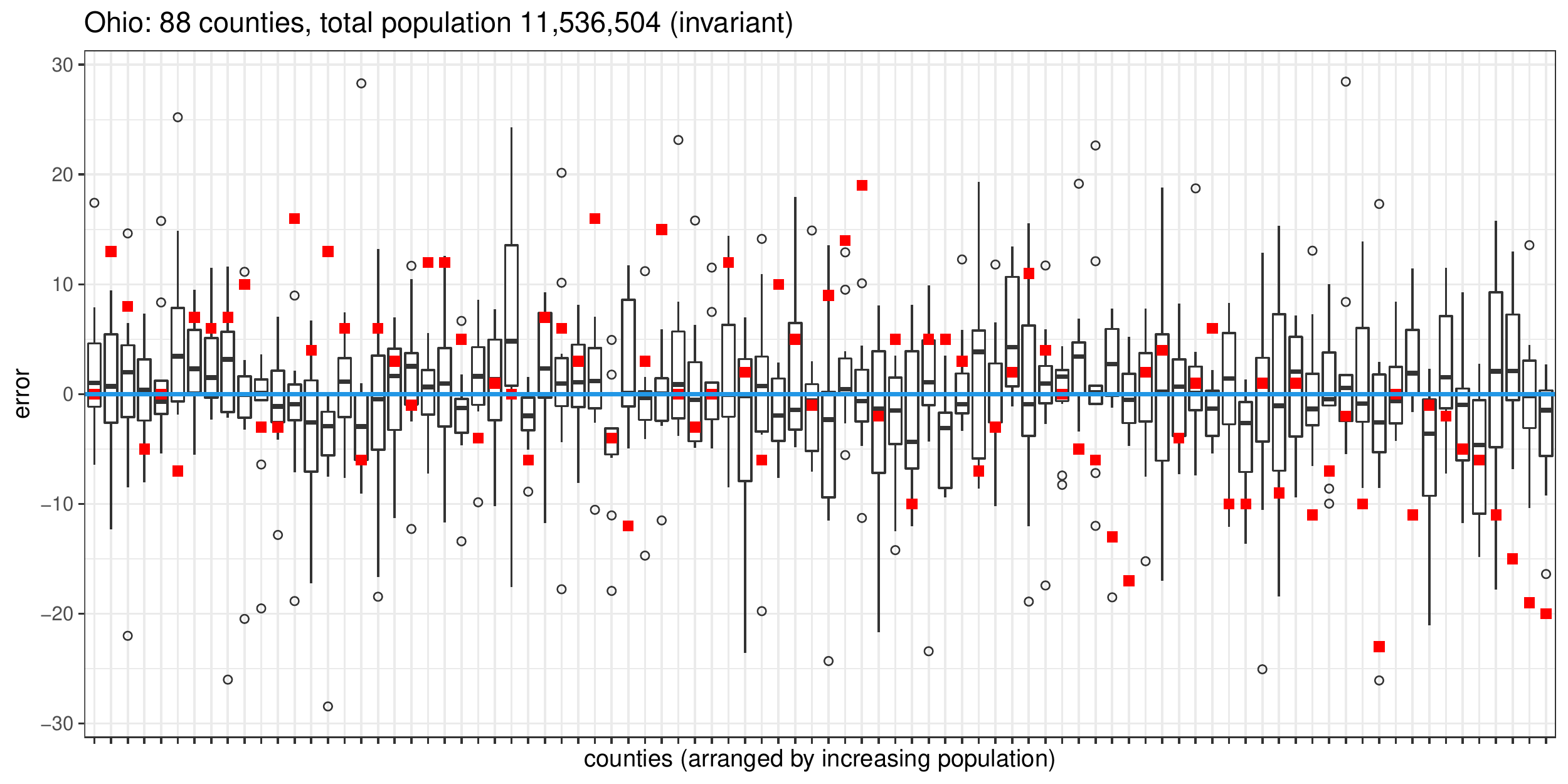}

\includegraphics[width = .47\textwidth]{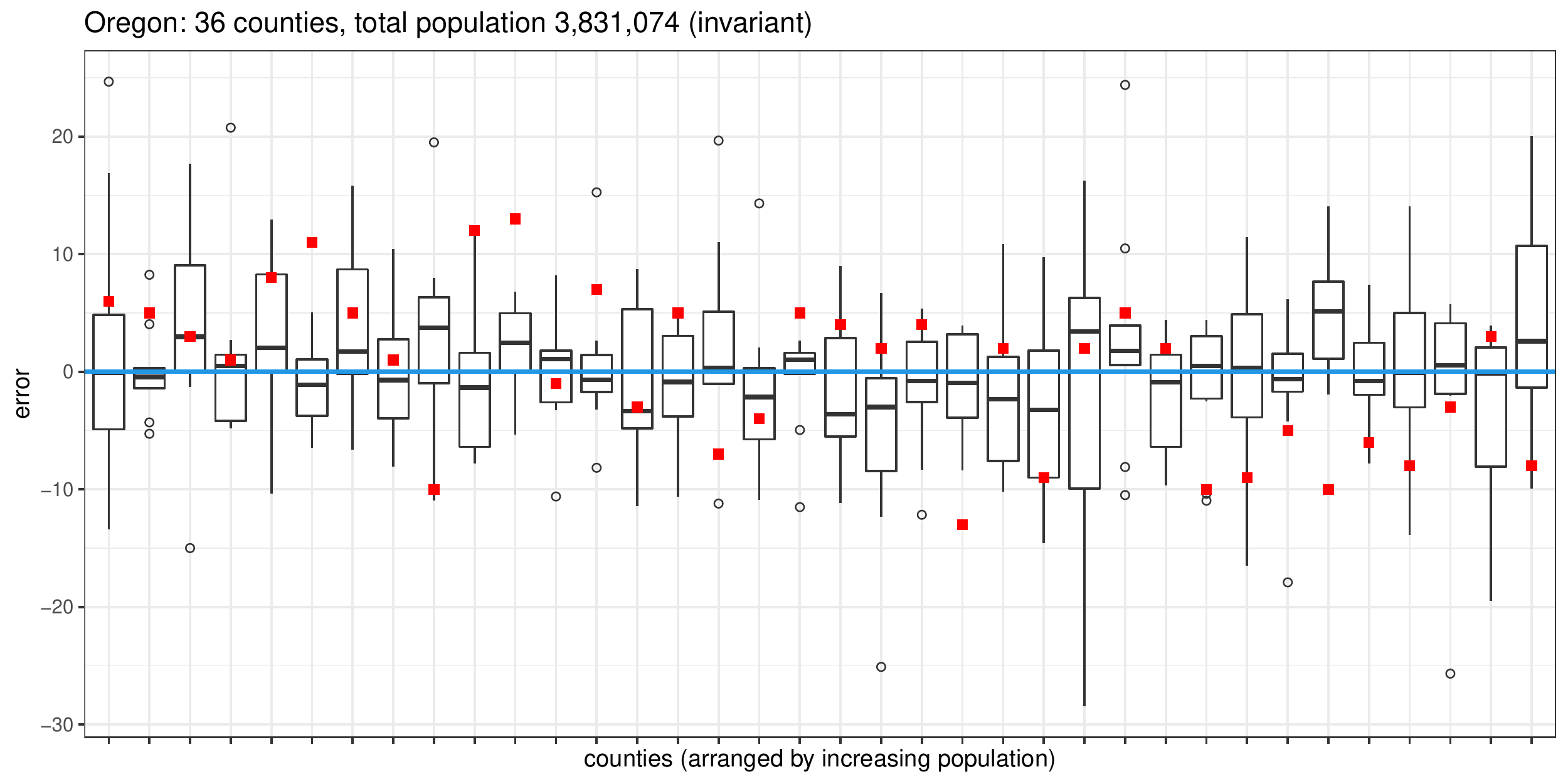}
\includegraphics[width = .47\textwidth]{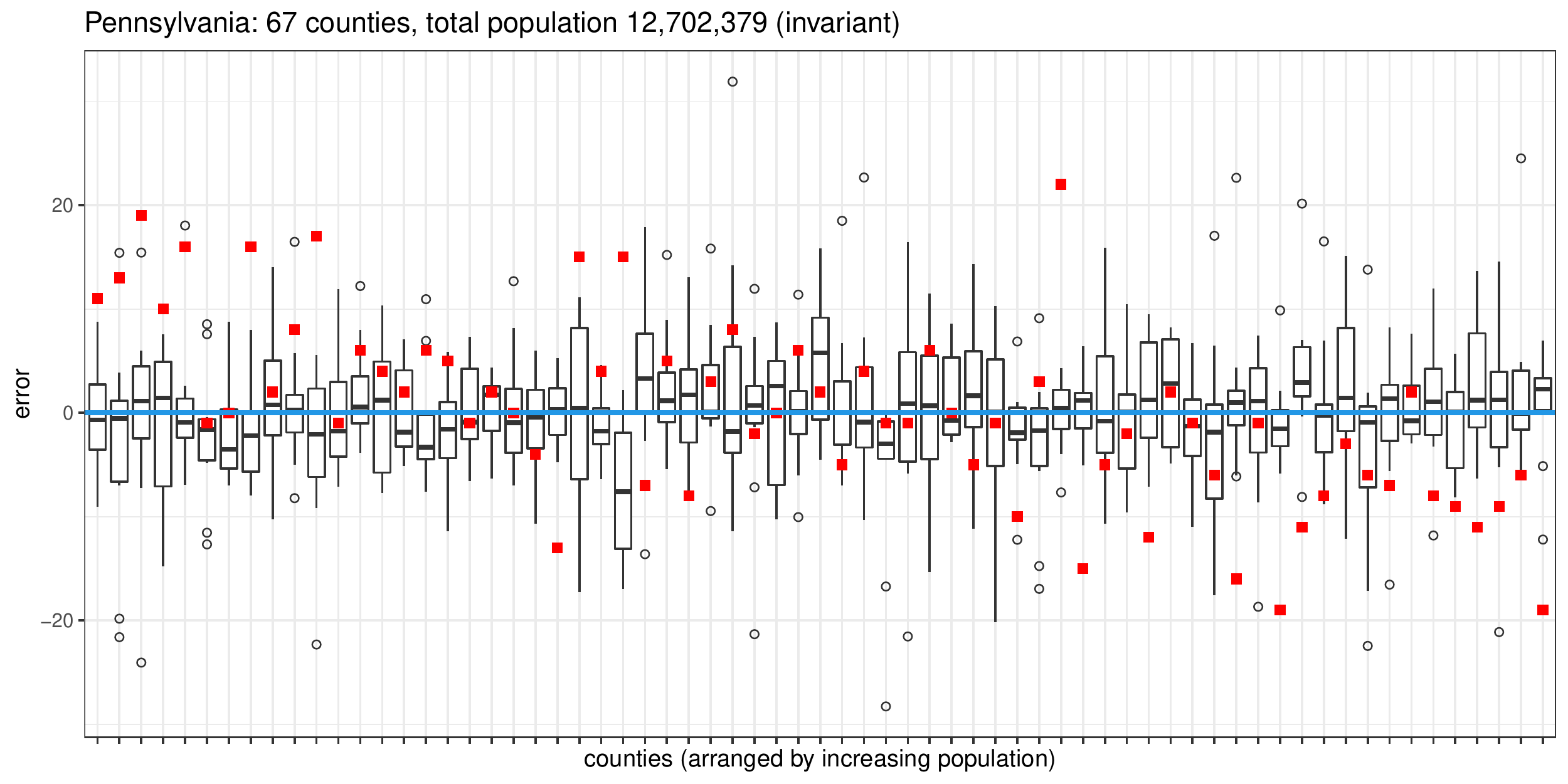}

\includegraphics[width = .47\textwidth]{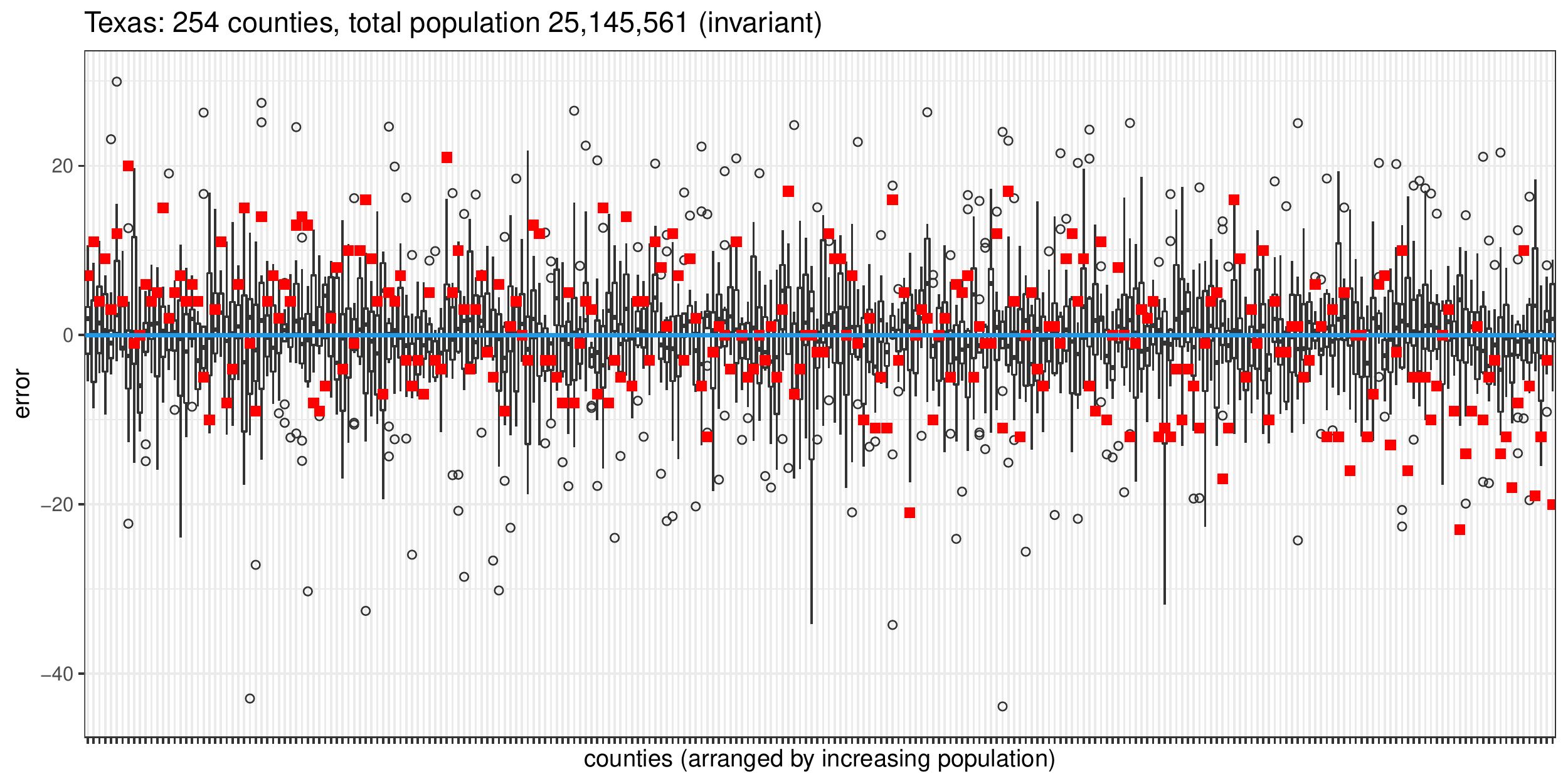}
\includegraphics[width = .47\textwidth]{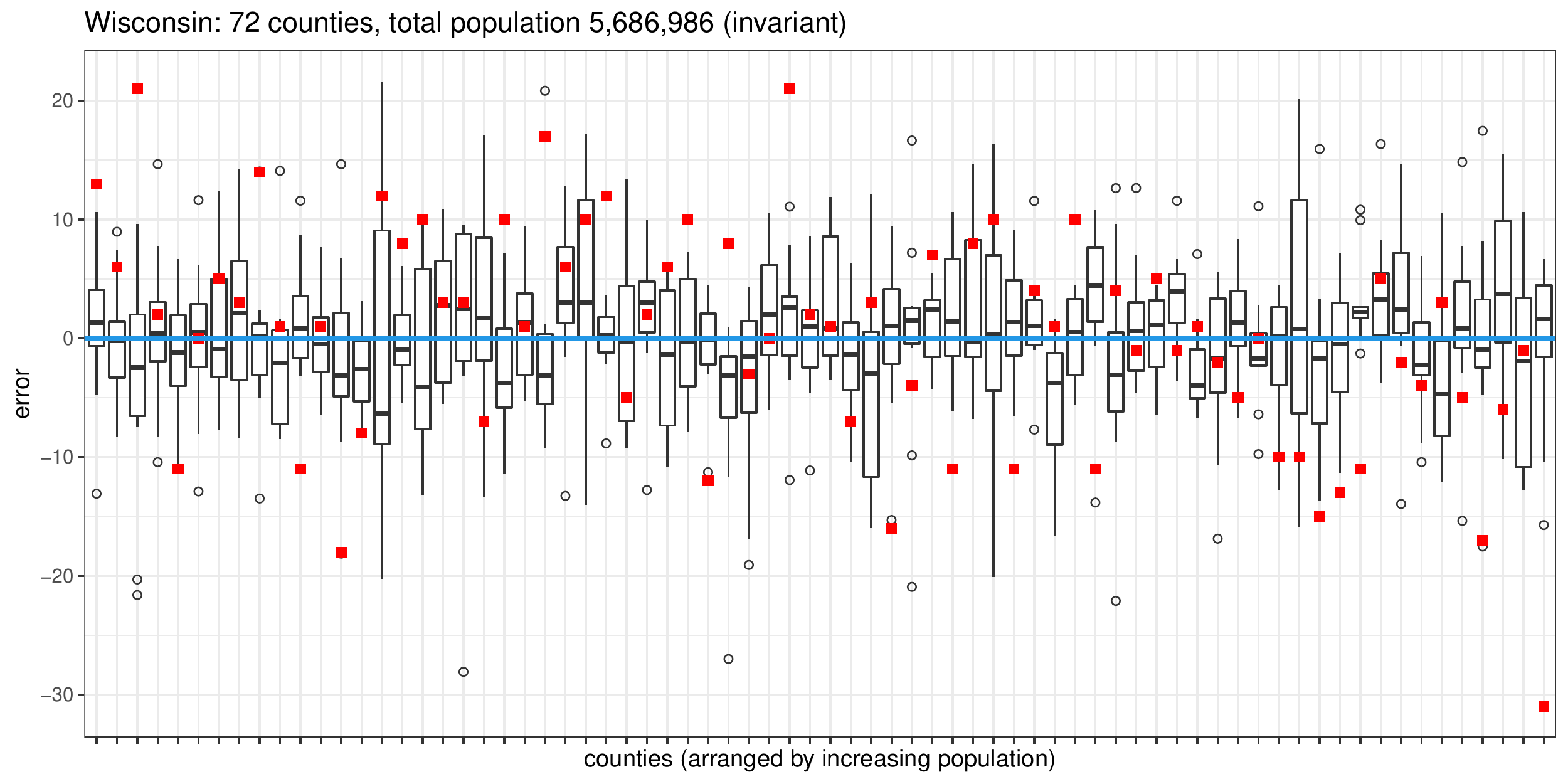}
\end{center}
\caption{\label{fig:census_sdp_county2}  Projected Laplace $(\epsilon,0)$-induced subspace differentially private mechanism of~\Cref{cor:projected_laplace} (boxplots; 10 runs) versus DAS errors (red squares), for states for which the TopDown algorithm incurred decidedly negatively biased errors in county-level counts. See Section~\ref{sec:intro}, Figure~\ref{fig:census_trend} and Section~\ref{sec:example} for more details).}	
\end{figure*}

\begin{sidewaysfigure}
    \centering
    \includegraphics[width = \textheight]{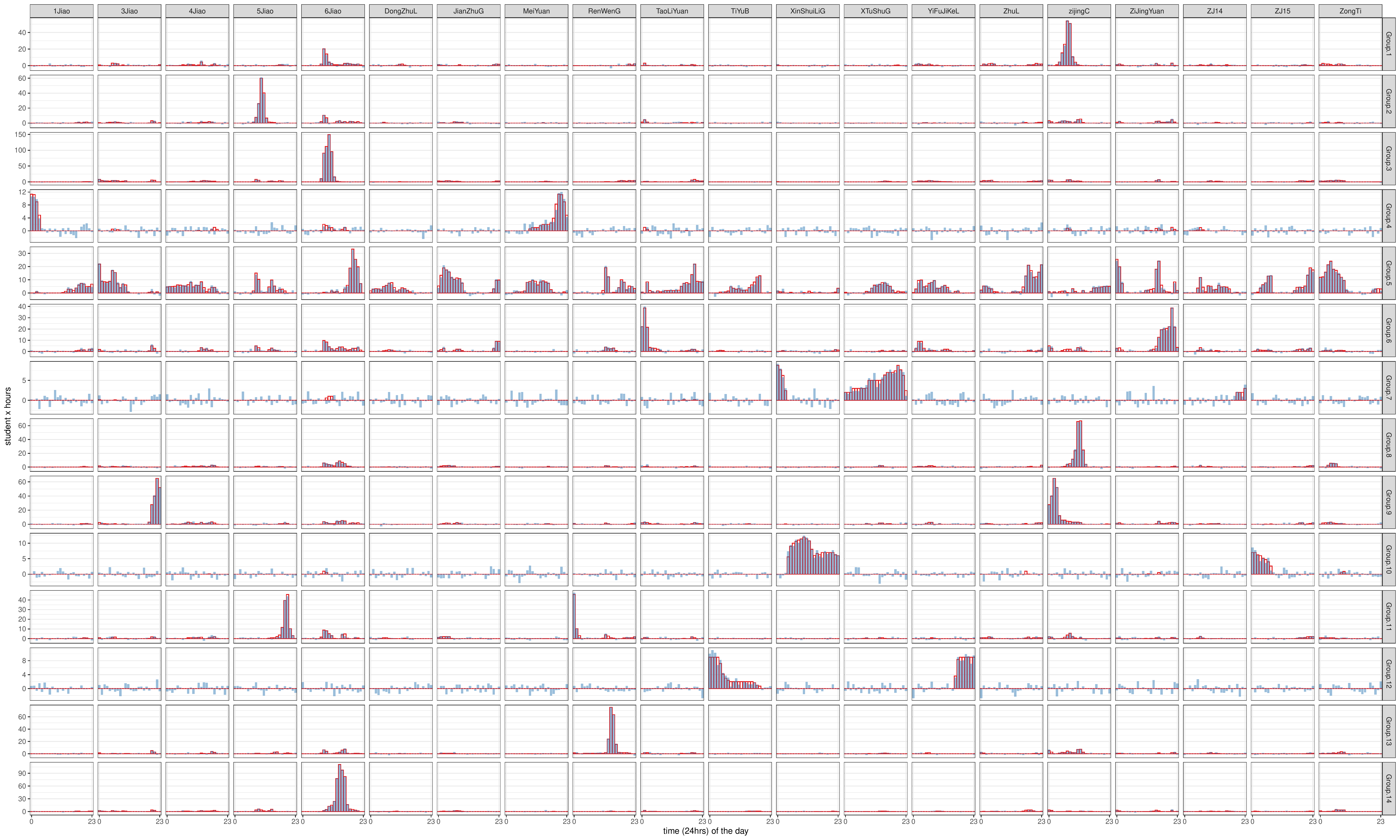}
    \caption{\label{fig:tsinghua} Hourly person-hour counts per group and building, subject to invariants 1) the total number of person-hours spent at each building every hour from all groups, and 2) the total number of person-hours spent at each building by every group over 24 hours. The spatio-temporal query is  $6720$ dimensional, subject to a $740$-dimensional linear constraint. Red histograms are the confidential query values, and blue histograms represent one run of the projected Gaussian induced subspace differential privacy mechanism of~\Cref{cor:projected_gaussian} with standard scale.}
\end{sidewaysfigure}

\end{document}